\documentclass[a4paper,UKenglish]{lipics-v2016}

\usepackage{microtype}

\usepackage{bbm}
\usepackage{amsmath}
\usepackage{amssymb}
\usepackage{mathrsfs}

\newtheorem{observation}[theorem]{Observation}
\newcommand{\MyHide}[1]{}

\title{Non-approximability and Polylogarithmic Approximations  of the Single-Sink Unsplittable and Confluent  Dynamic Flow Problems
\footnote{The work of all three authors  was partially supported by RGC Hong Kong CERG grant 16208415.}
}
\titlerunning{Non-approximability and Polylogarithmic Approximations  of Dynamic Flows} 

\author[1]{Mordecai J. Golin}
\author[2]{Hadi Khodabande}
\author[3]{Bo Qin}
\affil[1]{CSE Department. Hong Kong UST, \texttt{golin@cse.ust.hk}}
\affil[2]{
CE Department. Sharif University of Technology,
\texttt{khodabande@ce.sharif.edu}}
\affil[1]{CSE Department. Hong Kong UST,  \texttt{bqin@cse.ust.hk}}
\authorrunning{M. J. Golin, H. Khodabande and B. Qin} 

\Copyright{Mordecai J. Golin, Hadi Khodabande and Bo Qin}

\subjclass{G.1.6, G.2.1, G.2.2
\vspace{-0.5cm}
}
\keywords{Optimization, Approximation, Dynamic Flow, Confluent Flow
\vspace{-0.5cm}
}

\EventEditors{Yoshio Okamoto and Takeshi Tokuyama}
\EventNoEds{2}
\EventLongTitle{28th International Symposium on Algorithms and Computation (ISAAC 2017)}
\EventShortTitle{ISAAC 2017}
\EventAcronym{ISAAC}
\EventYear{2017}
\EventDate{December 9--12, 2017}
\EventLocation{Phuket, Thailand}
\EventLogo{}
\SeriesVolume{92}
\ArticleNo{41} 

\begin{document}

\maketitle
\vspace{-0.8cm}
\begin{abstract}
{\em Dynamic Flows} were introduced by Ford and Fulkerson in 1958 to model flows over time. They differ from standard network flows by defining edge  {\em capacities} to be the total amount of flow that can enter an edge {\em in one time unit}. In addition, each edge has a {\em length}, representing the time needed to traverse it.
Dynamic Flows  have  been used to model many problems including traffic congestion, hop-routing of packets and evacuation protocols in buildings. While the basic  problem of moving  the maximal amount of supplies from sources to sinks is polynomial time solvable, natural  minor modifications can make it NP-hard.
 One such  modification is  that flows be {\em confluent,} i.e., all flows leaving a vertex must leave along the same edge.  This corresponds to natural conditions in, e.g., evacuation planning and hop routing.  

We investigate  the {\em single-sink Confluent Quickest Flow} problem.  The input is a graph with edge capacities and lengths,  sources with supplies and a sink. The problem is to  find  a confluent flow minimizing the time required to send  supplies to the sink.  
Our main results include:
\begin{enumerate}

\item [$\bullet$]
{\em Logarithmic Non-Approximability.}
Directed Confluent Quickest Flows cannot be approximated in polynomial time with an $O(\log n)$ approximation factor, unless $P=NP$.

\item [$\bullet$]
{\em Polylogarithmic Bicriteria  Approximations.}
Polynomial time $(O(\log^8 n), O(\log^2 \kappa))$  bicritera approximation algorithms for the Confluent Quickest Flow problem where $\kappa$ is the number of sinks,  in both directed and undirected graphs.
\end{enumerate}
Corresponding results are also developed for the {\em Confluent Maximum Flow over time} problem. 
The techniques developed  are also  used to improve recent approximation algorithms for {\em static} confluent flows.

\vspace{-0.3cm}
\end{abstract}

\section{Introduction}
\label{Sec:Introduction}
Network Flow problems  
are very well known.  Their input is  a graph network with {\em capacities} $c(e)$  on its edges.  $c(e)$ is  the maximum  {\em flow} that can be pushed through $e$. The problem is usually  to maximize the amount of flow that can be pushed through the network.
By contrast, {\em Dynamic} network flows, while  introduced by Ford and Fulkerson \cite{Ford1958} in 1958, around the same time as regular network flows, are not
as well known.   In Dynamic Flows, $c(e)$ becomes  the amount of flow {\em that can enter $e$ in one time unit} 
while edge  {\em length}  $\ell(e)$ is the time that it takes for a unit of flow to traverse $e.$  Dynamic Flow problems need to consider the additional problem of {\em congestion}, which  may arise while flow waits to enter an edge.

Dynamic flows  have been used to model problems as diverse as traffic movement,  evacuation protocols and hop-routing of packets.  The {\em (Dynamic) Maximum Flow Over Time} problem  is  to find the maximum amount of flow that can be pushed from sources to sinks in a given amount of time. The  {\em (Dynamic) Quickest Flow} problem is     to find the  minimum time in which a fixed amount of flow can be pushed from sources to sinks.  In addition, there are  {\em multicommodity-flow} versions which require specific amounts of flow between given  source-sink pairs  and  {\em transshipment problems}  versions which
 do not restrict which source's demands are pushed to which sinks. It is known that the {\em Quickest Multicommodity Flow Over Time} problem is NP-Hard \cite{hall2007multicommodity}  while the {\em Quickest Transshipment} problem can be solved in polynomial time  \cite{Hoppe1994,Hoppe2000b}.
 Good surveys on Dynamic Flow problems and an introduction to its basic literature 
can be found in
\cite{kohler2009traffic,Pascoal2006,Skutella2009}.

In  basic (static)  network flow problems, {\em  splittable} flow is permitted, i.e., flow between a source and sink  can be divided into multiple parts with each  being routed over a different path.  
{\em Unsplittable flows} require that all flow between a particular source and sink be routed over only one path.  {\em Confluent flows} require that all flow passing through a vertex must leave that vertex on the same edge\footnote{Thus, confluent flows partition flows into edge disjoint in-trees, with the root of each tree being a sink.} \cite{chen2007almost,shepherd2015polylogarithmic}. Very recent work \cite{SV15}  has shown that, for the static single-sink case,  unless $P=NP$,  optimal unsplittable flows and optimal confluent flows do not have   polynomial time constant-factor approximation\footnote{The objectives studied  in \cite{SV15} are the total {\em amount} of flow that can be confluently routed or the number of demands that can be confluently satisfied in the {\em static} flow.} algorithms and, in fact,  confluent flows can not be approximated   to within a factor of $O(m^{1/2-\epsilon})$.

Confluent flows were  introduced by~\cite{merge03}, with  applications including  Internet routing~\cite{bley2008routing}, evacuation problems~\cite{Mamada2005}, 
and traffic coordination~\cite{kohler2009traffic}. 
Several works have studied confluent flows that minimize the maximum {\em congestion} in routing networks 
e.g.,~\cite{merge03,chen2007almost,shepherd2015polylogarithmic}. However, these works usually do not take into consideration the transit time (or edge length)  required for a packet  to traverse a single link,
though this parameter is usually considered in general network analyses  (see, e.g.,
~\cite{Harris2013routing}).
This immediately raises the Confluent Quickest Flow problem: Does there exists any routing scheme that minimizes the total time for sending all packets via a feasible (congestion bounded) confluent flow?

Another scenario in which  {\em confluent dynamic flows} arise
naturally is
in modelling evacuation protocols. 
Let vertices represent locations to be evacuated and edges represent paths between vertices.
 A vertex's original  supply  is  the number of people to be evacuated from it and a sink corresponds to an emergency exit.   $\ell(e)$  is  the time required to traverse  path $e$;  $c(e)$ is  the number of people  that can enter $e$ in parallel,  i.e., its width. The Confluent Flow restriction states that all people passing through a vertex   must leave by the same edge, i.e., following a sign pointing ``This way out''.  
The Quickest Flow problem corresponds  to placing the exit signs so as to minimize the time required to evacuate all people.  The Maximum Flow Over Time problem corresponds  to placing the signs so as to maximize the number of people that  can be evacuated in a given amount of time.


The  single-source single-sink version of the Confluent Quickest Flow problem is known as the {\em Quickest-Path Problem} and has long been known to be polynomial-time solvable 
\cite{Pascoal2006}. 
The Confluent Flow version of the multiple-source multiple-sink Quickest Transshipment problem was  known to be polynomial-time solvable when $G$ is a tree \cite{Mamada2005}. It was also known that, for general graphs, the single-sink Confluent Quickest Transshipment problem is NP-Hard \cite{KAMIYAMA2009}.  But  no other hardness complexity results, and in particular, non-approximability results, were known for 
general $G$.

Our first results  are that {\em Confluent Dynamic Flow} problems on directed graphs, both the Quickest Flow and Max Flow Over Time versions, cannot be approximated to within $O(\log n)$ ($n$ being the number of vertices in $G$) unless $P=NP$.
Our results hold even when the graph has {\em a single sink}.
Since, Multicommodity Flow and Transshipment are equivalent in the single-sink case we write  ``Quickest Flow'' instead of ``Quickest Multicommodity Flow'' or ``Quickest Transshipment''.


In the other direction,  we present  {\em polylogarithmic bicriteria  approximation} algorithms for both  the {\em single-sink} Confluent Quickest Flow and Confluent Maximum Flow Over Time problems, in both directed and undirected networks.
Note that known approximation algorithms for confluent flows are restricted to static networks in~\cite{merge03,chen2007almost, shepherd2015polylogarithmic}, and known optimal algorithms  for dynamic confluent flows are restricted to special graphs, e.g., trees~\cite{Mamada2005}. To the best of our knowledge, our algorithm is the first polylogarithmic approximation for these problems in general networks. These results are presented in Tables~\ref{tab-results1}-\ref{tab-results2}.

\vspace{-0.4cm}
\subsection{Single-Sink Dynamic Unsplittable/Confluent Flow Problems}
\label{subsec:Problem Types}
The input to the problems is a dynamic flow network, i.e.,
a graph 
$G=(V,E)$ with $n$ nodes and $m$ edges, where edge $e$ has  capacity $c(e)$ and length $\ell(e)$. Also specified are a collection of sources $\{s_1,...,s_k\}\subset V$ and  a sink $t \in V$.  The problems studied are: 
\begin{itemize}
\vspace{-0.cm}
\item [$\bullet$] { \sc Quickest Flow Problem:}  Provides   additional inputs  $\{d_1,...,d_{\kappa}\}$.   $d_i$ is the supply at source  $s_i$.  The problem is to find a flow minimizing  the time it takes to send all of the   $d_i$ units of supply to  sink $t.$  
\item [$\bullet$] { \sc Maximum Flow Over Time Problem:} Provides additional input of time horizon $T.$ 
The problem is to find  a flow maximizing the amount of supply sent to the sink $t$ within  time horizon $T$.  Supply at the $s_i$ is unlimited.
\vspace{-0.2cm}
\end{itemize}

We treat two different types of flow restrictions:
\begin{itemize}
\vspace{-0.cm} 
\item  [$\bullet$]    {\sl Unsplittable Flow:} All flow from $s_i$ to $t$ must pass along the same  path $P_i$
from $s_i$ to $t$.
\vspace{-0.cm}
\item [$\bullet$] {\sl Confluent Flow:} Any two supplies that meet at a node must traverse an identical path to the sink $t$. In particular, at most one edge out of each node $v$ is allowed to carry flow. Consequently, the support of the flow is a tree with all paths in the tree terminating at $t.$
\vspace{-0.cm}

\end{itemize}

\begin{table*}[t]
  \centering

\begin{footnotesize}
\begin{tabular}{|l|l|l|}

	\hline

	 {\scriptsize Flow} &  {\scriptsize Dynamic Network }  &  {\scriptsize Hardness or LB on Approx. Ratio} \\

\hline
	\hline

{\scriptsize Confluent} & {\scriptsize Trees }  & {\scriptsize Polynomial-Time Solvable~\cite{Mamada2005} }\\

\hline

{\scriptsize Confluent }& {\scriptsize Directed/Undirected }  & {\scriptsize NP-Hard~\cite{KAMIYAMA2009}}\\

	\hline

{\scriptsize Unsplittable} & {\scriptsize Directed/Undirected }   & {\scriptsize $3/2-\epsilon$ (Thm.~\ref{thm-2approfixed}) }\\
\hline

{\scriptsize Confluent} & {\scriptsize Directed}   & {\scriptsize $\Omega(\log n)$ (Thm.~\ref{thm-conlogappro})*} \\
\hline

{\scriptsize Unsplittable/Confluent} & {\scriptsize Directed/Undirected }  & {\scriptsize  No $(\frac{15}{14}-\epsilon,1+\alpha)$-Approx. (Thm.~\ref{thm-biapproxhardQuick})*}\\
	\hline

  \end{tabular}
  \end{footnotesize}
\renewcommand{\captionfont}{\scriptsize}
\renewcommand{\tablename}{*}
\vspace{-0.cm}
\caption*{* Corresponding results also hold for the single-sink Maximum Flow Over Time problem (Thms~\ref{thm-unspovertime},~\ref{thm-confluentovertime},~\ref{thm-biapproxhardTime}).
}
\vspace{-0.3cm}
\renewcommand{\captionfont}{\small}
\renewcommand{\tablename}{Table}
\caption{Hardness or lower bounds on approx. ratio for the single-sink Quickest Flow problem.}

  \label{tab-results1}
\vspace{-0.5cm}
\end{table*}


\begin{table*}[t]
  \centering
\begin{footnotesize}
  \begin{tabular}{|l|l|l|l|l|l|}
	\hline

	{\scriptsize Network } & {\scriptsize Capacity} & {\scriptsize Objective} & {\scriptsize Sources}  & {\scriptsize Sinks} & {\scriptsize UB on Approx. Ratio} \\

\hline
\hline

{\scriptsize Static} & {\scriptsize Uncapacitated} & {\scriptsize Min Congestion*} & {\scriptsize $n$}  & {\scriptsize $k \ (k\leq n)$} & {\scriptsize $O(\log^3 n)$~\cite{merge03}$\dagger$} \\
\hline

{\scriptsize Static} & {\scriptsize Uncapacitated} & {\scriptsize Min Congestion*} & {\scriptsize $n$}  & {\scriptsize $k \ (k\leq n)$} & {\scriptsize $1+\ln k$~\cite{chen2007almost}} \\
\hline

{\scriptsize Static} & {\scriptsize Uncapacitated} & {\scriptsize Min Congestion*} & {\scriptsize $\kappa\ (\kappa\leq n)$}  & {\scriptsize $k \ (k\leq n)$} & {\scriptsize $O(\log^3 \kappa)$ (Thm.~\ref{thm-mergetree})$\dagger$}\\
\hline

{\scriptsize Static} & {\scriptsize Node} & {\scriptsize Max Demand} & {\scriptsize $n$}  & {\scriptsize $1$} & {\scriptsize $O(\log^6 n)$ with NBA$^{\mbox{\tiny\ref{NBA}}}$~\cite{shepherd2015polylogarithmic}} \\
\hline

{\scriptsize Static} & {\scriptsize Edge/Node} & {\scriptsize Max Demand}  & {\scriptsize $\kappa\ (\kappa\leq n)$ }  & {\scriptsize $1$} & {\scriptsize $O(\log^{10} \kappa)$ with NBA$^{\mbox{\tiny\ref{NBA}}}$ (Thm.~\ref{thm-StaticConfluent})$\dagger$}\\
\hline

{\scriptsize Dynamic} & {\scriptsize Edge} & {\scriptsize Max Flow Over Time} & {\scriptsize $\kappa \ (\kappa\leq n)$}  & {\scriptsize $1$} & {\scriptsize $(O(\log^2 \kappa),O(\log^8 n))$ (Thm.~\ref{thm-PolylogConFlowOverTime})$\dagger$}\\

\hline

{\scriptsize Dynamic} & {\scriptsize Edge} & {\scriptsize Quickest Flow} & {\scriptsize $\kappa \ (\kappa\leq n)$}  & {\scriptsize $1$} & {\scriptsize $(O(\log^8 n), O(\log^2 \kappa))$ (Thm.~\ref{thm-PolylogConQuickFlow})$\dagger$}\\

\hline

  \end{tabular}
  \end{footnotesize}
\renewcommand{\captionfont}{\scriptsize}
\renewcommand{\tablename}{*}
\vspace{-0.1cm}\caption*{*
Minimize the maximum node congestion in a network that admits a feasible splittable flow satisfying all supplies.
}
\renewcommand{\tablename}{$\dagger$}
\vspace{-0.3cm}
\caption  *{ \hspace{-.1cm} $\dagger$ These results hold with high probability, or more precisely, with probability $1-n^{-c}$, where $c$ is a constant.
}
\vspace{-0.2cm}
\renewcommand{\captionfont}{\small}
\renewcommand{\tablename}{Table}
\caption{Upper bounds on approximation ratio for variations of the Fixed-Sink Confluent Flow problem. 
The first three  items are for uncapacitated problems but are included here because they serve as the internal building blocks for the approximation algorithms for the capacitated problems.}

\label{tab-results2}
\vspace{-0.8cm}
\end{table*}

\vspace{-0.4cm}
\subsection{Our Results}
\label{subsec: our results}
Section \ref{Sec:2_Approx}
presents a simple proof that, unless $P=NP$,   $\forall \epsilon >0,$ it is impossible to construct a polynomial-time $3/2-\epsilon$ approximation algorithm
for the single-sink Quickest Flow problem when flows are restricted to be either unsplittable or confluent. This result 
holds
 for both directed and undirected graphs and even when the graph is  restricted to have  only one sink

Section \ref{Sec:Log_Hardness} proves, for the confluent directed graph case,   the much stronger result  that unless $P=NP$,
it is impossible to construct a polynomial-time
$O(\log n)$   approximation algorithm
for the single-sink Quickest Flow problem.  
 The major tool used is a modification of a grid graph construction from \cite{SV15}  which was an extension of one pioneered by~\cite{G2003}.
We note that our reduction is not the same as that in  \cite{SV15}. There,  the objective function was the maximum {\em amount} of static flow that could be pushed. Here, the objective function is the minimum amount of {\em time}  required to push the supplies. 
Our proof works by deriving new properties of the grid-graph. Section \ref{Sec:Max_Flow_Over_Time} extends the analysis to the Maximum Flow Over Time problem with our
 lower bounds on the approximation ratio being summarized in
Table~1.

We also note that it might seem intuitive that, because confluent flows are ``harder'' than static flows,  the non-approximability of confluent static flows, e.g.,  the result from  \cite{SV15}, should immediately imply the non-approximability of confluent dynamic flows.  This is not true, though.  The two problems are trying to optimize very different things, making them incomparable.  More specifically, in the static case, the goal is  {\em Demand Maximization}, i.e., to find a subset of the demands of maximum total value that can be confluently routed.  In the dynamic case, the goal is to find  a confluent routing of ALL demands in minimal time.  To appreciate the distinction it is instructive to examine  confluent routing on {\em trees}  where  the static problem is
NP-Hard 
\cite{Dressler2014} 
but the dynamic case is {\em polynomial-time} solvable~\cite{Mamada2005}.

Despite the non-approximability shown above for confluent dynamic flows, one might hope to create  {\em bicriteria $(\alpha, \beta)$ approximation}s\footnote{These will be formally introduced in  Definition \ref{def:bic}.}. 
However, in Section~\ref{Sec:Constant_Bi_Approximation}, we demonstrate that, for both directed and undirected graphs,  there exists   a constant $\alpha > 0$ such that, for any $\epsilon>0$, there is no polynomial-time  $(\frac{15}{14}-\epsilon,1+\alpha)$-approximation for the Unsplittable/Confluent Quickest Flow problem, unless $P=NP$. Similar results are obtained for the Unsplittable/Confluent Max Flow Over Time problem. Our proof
utilizes a reduction from  the Bounded Occurrence 3-Dimensional Matching problem.

In contrast to the above we show, in Section~\ref{Sec:Polylog_Approx_Confluent}, how to construct  a $(O(\log^8 n),O(\log^2 \kappa))$-approximation for  the Confluent Quickest Flow problem, where $\kappa$ is now the number of sources,  in  polynomial time. To this end, we use the idea of routing a confluent flow in a static {\em monotonic network}, i.e., one in which  each vertex is given an additional  {\em vertex capacity}  that satisfies that all  edges go from  a low-capacity node to a high-capacity one, which was introduced in \cite{shepherd2015polylogarithmic}.
  Recall that in our original confluent flow problem the support of the flow is  a tree.  In that tree, a parent node never supports less flow than its child.  So, intuitively,  a feasible confluent flow requires its tree support to be monotonic.
We develop new techniques (Theorem~\ref{thm-monotonicflow}) that permit constructing,
 in polynomial time, a confluent flow that routes all supplies in a given monotonic network, while bounding both {\em node congestion} and {\em flow length}.

Via this  monotonic technique, we build a novel  multi-layer monotonic network and construct a confluent static flow on it which is finally re-routed to produce a confluent dynamic flow for our original graph problem.
Our method guarantees that a {\em dynamic flow} can be found such that the total transit time is at most polylogarithmic factor times the optimal. Similarly,
this also lets us  develop a polynomial-time $(O(\log^2 \kappa),O(\log^8 n))$-approximation of the Confluent Maximum  Flow Over Time problem.

Our technique mainly differs from that in~\cite{shepherd2015polylogarithmic}  in constructing  {\em length-bounded} confluent flows in {\em static} networks (which might be of independent interest). It also permits us to improve their approximation algorithms when not all vertices are sources.
More specifically,  recall that  \cite{shepherd2015polylogarithmic} gives an
$O(\log^6 n)$
approximation algorithm for the demand maximization confluent
flow problem, with the no-bottleneck assumption (NBA)\footnote{\label{NBA}In node-/edge-capacitated networks, the NBA is that  $\max_{v\in V} d(v)\leq \min_{v\in V}c(v)$, and $\max_{v\in V} d(v)\leq \min_{e\in E}c(e)$, resp..}. If restricted to static networks, our technique can give an $O(\log^{10} \kappa)$
approximation for the same problem.  If $\kappa$ is bounded, for example, this gives a {\em constant} approximation, which is nearly optimal.

Our improvement to the approximation ratio comes through a combination of   (i) a novel construction of the multi-layer network, and (ii) a new building block inside our monotonic network technique---a better routing approach for {\em uncapacitated} networks (Theorem~\ref{thm-mergetree}). This will be discussed in more detail in  Section~\ref{Sec:Polylog_Approx_Confluent}.


Our  Theorem~\ref{thm-mergetree} enables us to route confluent flows in uncapacitated monotonic sub-networks with congestion bounded by poly$(\log \kappa)$ instead of poly$(\log n)$.
While this might look weak compared to the $1 + \ln k$ ($k$ being the number of sinks) bound from \cite{chen2007almost} this is only used as a subroutine.  In fact,  the internal constructions of both \cite{shepherd2015polylogarithmic}  and our proofs for approximating the {\em capacitated} static problem build uncapacitated sub-networks which can have $\Theta(n)$ induced sources and sinks.  Plugging in the bound of~\cite{chen2007almost} would give a poly$(\log n)$ bound.  We develop a new  combinatorial argument that, combined with our new poly$(\log \kappa)$ bounds for  uncapacitated monotonic sub-networks, gives a poly$(\log \kappa)$ bound for the capacitated one as well, yielding our Theorem~\ref{thm-monotonicflow-relaxed}. This leads us to the final improvement.

A chart  presenting previously known results and our new ones is given in
 Table~\ref{tab-results2}.

\vspace{-0.2cm}
\section{Preliminaries: Definitions and NP-Hard Problems}
\label{Sec:Preliminaries}

Let $I$ be some input to an optimization problem,  $OPT(I)$ be the {\em optimum} value to the given problem on $I$ and $|I|$ be its {\em size}.
As examples,
$I$ could be a dynamic flow problem on a graph with $n$ vertices  and $m$ edges. We could have just as easily defined $|I|= m+n$.

\MyHide{Let ${\cal A}$  be some algorithm for approximately solving a given optimization problem with $A(I)$ being the value of the solution produced by ${\cal A}$ on input instance $I$.
For $\rho > 0$, we call ${\cal A}$ as a {\em $\rho$-approximation algorithm} if either (i) for all instances $I$ in the minimization problem, $A(I) \le \rho \cdot OPT(I)$, or (ii) for all instances $I$ in the maximization problem,  $\rho \cdot A(I) \ge  OPT(I)$.}

We now define {\em bicriteria approximation}s for the two-objective optimization problem.

\vspace{-0.2cm}\begin{definition}[Bicriteria Approximation]
For any $\alpha, \beta>0$, an $(\alpha, \beta)$-approximation algorithm ${\cal A}$ for the two-objective optimization problem is a function that takes as input any parameter $k$ and any instance $I$,  and outputs a solution $x$ such that
\begin{enumerate}
\item  $\alpha f(x)\geq f(x^*),\ g(x)\leq \beta k$, if the optimization problem is to find a solution $x$ maximizing the cost function $f(x)$ subject to another cost function $g(x)\leq k$,
\item  $f(x)\leq \alpha  f(x^*),\ \beta g(x)\geq  k$, if the optimization problem is to find a solution $x$ minimizing the cost function $f(x)$ subject to another cost function $g(x)\geq k$,
\end{enumerate}
where $x^*$ is the optimal solution for the input $I$ and $k$.
%
\label{def:bic}
\end{definition}

We can actually define two different types of confluent flows:
\vspace{-0.2cm}\begin{definition}
A flow in $G$ is   {\em node-confluent} if, for every  vertex $v$, all flow leaving $v$ leaves along the same edge.
A flow in $G$ is   {\em edge-confluent} if, for every  edge $e=(u,v)$ if all flow that passes through $e$ must leave $v$ through the same edge $(v,w)$.
\end{definition}
\vspace{-0.2cm}
In this paper the term ``confluent'', when used alone, will denote node-confluence.  When edge-confluence is needed (in some proofs) it will be explicitly specified.

Finally we will use  the following NP-hard problems in  our reductions:


\vspace{-0.2cm}\begin{definition}{\bf The Two-Disjoint Paths (Uncapacitated) Problem:}
Given a graph $G$ and node pairs $\{x_1, y_1\}$ and $\{x_2, y_2\}$,
decide if $G$ contains paths $P_1$ from $x_1$ to $y_1$ and $P_2$ from $x_2$ to $y_2$ such that they are disjoint.
\end{definition}\vspace{-0.2cm}
In undirected graphs the Two-Disjoint Paths (Uncapacitated) problem, for both edge-disjoint and node-disjoint paths, is polynomial-time solvable~\cite{Robertson}. However, in directed graphs, the problem is NP-hard for both edge-disjoint and node-disjoint paths~\cite{FHW}.

\vspace{-0.1cm}\begin{definition}{\bf The Two-Disjoint Paths (Capacitated) Problem:}
 Let $G$ be a (static) graph whose edges are labelled   either $\alpha$ or $\beta$ with $\beta\geq \alpha$. These labels are the  {\em capacities} of the edges. Given node pairs $\{x_1, y_1\}$ and $\{x_2, y_2\}$, decide whether $G$ contains paths $P_1$ from $x_1$ to $y_1$ and $P_2$ from $x_2$ to $y_2$ such that:
\begin{itemize}
\item[i.] $P_1$ and $P_2$ are disjoint (node-disjoint or edge-disjoint);
\item[ii.] $P_2$ may only use edges of capacity $\beta$ ($P_1$ may use both capacity $\alpha$ and capacity $\beta$ edges).
\end{itemize}
\end{definition}
The version of node-disjoint  paths  was proven to  be NP-hard for undirected graphs by~\cite{G2003}.
The version of edge-disjoint paths was proven to be NP-hard by~\cite{NSV10}.

\vspace{-0.2cm}\begin{definition}[The Bounded Occurrence 3-Dimensional Matching Problem (BO3DM)]
Suppose there are three disjoint sets $A=\{a_1,...,a_n\}$, $B=\{b_1,...,b_n\}$
and $C=\{c_1,..,c_n\}$, and a set $T=\{T_{\mu}\in A\times B\times C:\mu\in[m]\}$ such that each element of $A, B, C$ occurs in the same constant number $M$ of triples in $T$. The goal is to find the largest subset $T'\subset T$ such that all triples in $T'$ are disjoint, i.e., no two elements of $T'$ contain the same element of $A,B,C$.
\end{definition}
\vspace{-0.2cm}
\cite{G14} shows that  there exists an $\epsilon_0>0$ such that it is NP-hard to
decide whether there exist $n$ disjoint triples in $T$ ({\em satisfiable} instance) or there
exist at most $(1-\epsilon_0)n$ disjoint triples in $T$ ({\em $\epsilon_0$-unsatisfied} instance).



{\em Dynamic Flows.} We first describe the mechanics of flow over one edge $e=(u,v)$ with capacity $c$ and length $\ell$.  Suppose there  are  $d$ units of supply on node $u.$  Assume the discrete case in which $d,c,\ell$ are all integral and all $d$ need to be moved from $u$ to $v$.  Items move in groups of size at most $c$, with one group entering $e$ each time unit.  Thus, the items are transported in $\lceil d/c\rceil$ groups.  It takes $\ell$ time units for the first group  to arrive at $v$.  Since the groups left $u$ at consecutive time units they arrive at $v$ in consecutive time units.  Thus, it requires
$\lceil d/c\rceil - 1 + \ell$ time to move all items from $u$ to $v$ over $e.$
Also, in both cases, if other items arrived at $u$ wanting to enter $e$ they would have to wait until all items already at $u$ had departed before entering $e$. To provide intuition, we  give examples of quickest flows in the unsplittable/confluent cases in Figure \ref{fig:Types of Flow}.

Finally, we introduce some notations. A flow $f$ is {\em feasible} if $\forall e \in E$, $f(e)\leq c(e)$. For any $e\in E$, we define its {\em edge congestion} as $EC(e):=f(e)/c(e)$. Under certain circumstance, we may introduce the {\em node capacity} $c(v)$ of $v\in V$, and define its {\em node congestion} $NC(v):=f^{out}(v)/c(v)$, where $f^{out}(v)$ is the total flow out of $v$. For a flow $f$, we let its edge congestion $EC(f):=\max_{e\in E}EC(e)$ and node congestion $NC(f):=\max_{v\in V\setminus \{t_1,...,t_k\}}NC(v)$, where $t_1,...,t_k$ are sinks.

A static flow $f$ can be specified by a collection of source-sink paths $\mathcal{P}=(P_1,...,P_K)$ and corresponding flow values $f_1,...,f_K$. We define the {\em length} of flow $f$ as $L(f):=\max_{i\in[k]}L(P_i)$, where $L(P_i):=\sum_{e\in P_i}\ell(e)$ is the length of $P_i$. $f$ is called as {\em $L$-length-bounded} for some $L\in\mathbb{R}^+$ if $L(f)\leq L$, i.e., no
path in $\mathcal{P}$ has path length longer than $L$. Also, if all $f_i$'s are identical, we call $f$ as {\em uniform}.


\section{Approximation Hardness for Unsplittable/Confluent Dynamic Flows}
\label{Sec:Approx_Hardness}

\subsection{Constant Approximation Hardness of the Quickest Flows Problem}
\label{Sec:2_Approx}

This section gives a simple proof  that a  polynomial-time constant approximation algorithm for the single-sink  Unsplittable/Confluent Quickest Flow problem would imply  $P=NP$  (proof in
Appendix~\ref{apd:Proof_of_Apd_2_Approx}).
\MyHide{ These results  
hold
in both Fixed-Sink and Sink-Location versions (proofs in
Appendix~\ref{apd:Proof_of_Apd_2_Approx}).}

\begin{theorem}\label{thm-2approfixed}
The  single-sink Unsplittable/Confluent Quickest Flow problem in both directed and undirected graphs
cannot be approximated to within a factor $3/2-\epsilon$, for any $\epsilon>0$, unless $P=NP$.
\end{theorem}

\vspace{-0.5cm}
\subsection{Logarithmic Approximation Hardness of Confluent Quickest Flows}
\label{Sec:Log_Hardness}

For the single-sink directed Confluent Quickest Flow problem we now derive a much stronger result   than in the previous section. That is, it is
 NP-hard to even  get a $O(\log n)$ approximation to the optimal solution.


To prove the logarithmic approximation hardness, we construct the following  instance.

{\bf Hard instance.}
Before building the desired  hard instance, we describe the dynamic half-grid network $G_N$. It can be viewed as an extension of the static half-grid graph in \cite{SV15}. There are $N$ rows (numbered from bottom to top) and $N$ columns (numbered from right to left). All the edges in the $i$-th row and all the edges in the $i$-th column have capacity $1/i$. The $i$-th row extends as far as the $i$-th column and vice versa. The sink $t$, located at the bottom of the half-grid, is connected with the bottom node $t_i$ of the $i$-th column by an edge of capacity $1/i$. Also, at the leftmost node of the $i$-th row, there is a source $s_i$ with supply $M^2/i$, where $M$ is a sufficiently large constant.  We set all edge lengths as $1$, and always enforce edge directions to
be downwards and to the right. The half-grid is given in Figure~\ref{fig:half-grid}.

Suppose we are now given an instance  $\mathcal{I}$ of the directed node-disjoint version of the Two-Disjoint Paths (Uncapacitated) problem. We replace each 4-degree node in the half-grid by a copy of $\mathcal{I}$. Inside the copy, all edges   have length $1$. Consider the copy of $\mathcal{I}$  at the
intersection of the $i$-th column and $j$-th row (with $j>i$) in $G_N$. That instance is incident to two edges of capacity $1/i$
and two edges of capacity $1/j$.
Inside that $\mathcal{I}$, we let the edges of capacity $1/j$ be incident
to $x_1$ and $y_1$, and the edges of capacity $1/i$ be incident to $x_2$ and $y_2$; we set all edge capacities to  $1/i$. This completes the hard instance of directed confluent dynamic flows. Denote the constructed network as $\mathcal{G}$.

Utilizing $\mathcal{G}$, we obtain the logarithmic approximation hardness for the Confluent Quickest Flow problem. The proof (Appendix~\ref{App:ApdLogHardness}) works by showing that if we could get a logarithmic approximation, we could solve $\mathcal{I}$. 

\begin{theorem}\label{thm-conlogappro}
The single-sink Confluent Quickest Flow problem in directed graphs
cannot be approximated to a factor within $O(\log n)$,

unless $P=NP$.
\end{theorem}



\vspace{-0.4cm}
\subsection{Approximation Hardness of the Max  Flow Over Time Problem}
\label{Sec:Max_Flow_Over_Time}
This section discusses the approximation hardness of the single-sink Unsplittable and Confluent Maximum Flow Over Time problem.  


To derive the approximation hardness of the  Unsplittable Maximum Flow Over Time problem, we will again reduce from the directed/undirected edge-disjoint version of Two-Disjoint Paths (Capacitated) problem.
We construct the same network as in Section~\ref{Sec:2_Approx} and utilizing this constructed network, we show (see also proofs in Appendix~\ref{App: Apd MFT})

\begin{theorem}\label{thm-unspovertime}
The single-sink  Unsplittable Maximum Flow Over Time problem in both directed and undirected graphs
cannot be approximated to  a factor within $3/2-\epsilon$, for any $\epsilon> 0$, unless $P=NP$.
\end{theorem}\vspace{-0.2cm}

Although the above hard instance applies to the confluent flow, we present a stronger lower bound for the Confluent Maximum  Flow Over Time in directed graphs.
\begin{theorem}\label{thm-confluentovertime}
The single-sink Confluent Maximum  Flow Over Time problem in directed graphs
cannot be approximated to  a factor within $O(\log n)$,
unless $P=NP$.
\end{theorem}



\section{Constant Bicriteria Approximation Hardness of Dynamic Flows}
\label{Sec:Constant_Bi_Approximation}

This section first proves the NP-hardness of  constant bicriteria approximations for the Unsplittable and Confluent Maximum Flow Over Time problems.

Our proof uses reductions from  the BO3DM problem.
Inspired by the reduction\footnote{Even though we are reducing to the same problem note that our goal differs from~\cite{G2003}, which
aims at finding a maximum number of length-bounded edge-disjoint paths.  For technical reasons, this requires us to develop a totally different bounding technique.}
presented in~\cite{G2003, G21}, given an instance of BO3DM, we construct the following corresponding hard instance for the Unsplittable/Confluent Maximum Flow Over Time problem in undirected graphs. Note that the directed case is similar, except that we enforce all edge directions to point  right.  Suppose we are given an instance $\mathcal{I}$ of Bounded Occurrence 3-Dimensional Matching problem. Denote the $\mu$-th triple $T_{\mu}$ as $(a_{p_{\mu}},b_{q_{\mu}},c_{r_{\mu}})$, where $p_{\mu},q_{\mu},r_{\mu}\in[n]$. We build an undirected graph $G=(V,E)$ (shown in Figure~\ref{fig:bicriteria}), where
\vspace{-0.2cm}
\begin{eqnarray}
 V&=&\{s,t\}\cup\{a_{il}:i\in[n],l\in[M-1]\}\cup\{s_i,b_i,c_i:i\in[n]\}\cup\{s'_{\mu},x_{\mu},y_{\mu}:\mu\in[m]\},\nonumber\\
 E&=&\{(s_i,s),(s,b_i),(c_i,t),(a_{il},t):i\in[n],l\in[M-1]\}\nonumber\\
 &&\cup\{(s'_{\mu},s),(s,x_{\mu}),(y_{\mu},a_{p_{\mu}l}):\mu\in[m],l\in[M-1]\}\nonumber\\
&&\cup\{(b_{q_{\mu}},x_{\mu}),(x_{\mu},y_{\mu}),(y_{\mu},c_{r_{\mu}}):\mu\in[m]\}.\nonumber
\end{eqnarray}\vspace{-0.66cm}

Hereby, $G$ contains a vertex representing each element in the sets $B$ and $C$, and $(M-1)$ copies of each element in $A$. Also, $G$ contains a sink $t$, and sources $s_i\ (i\in[n])$, $s'_{\mu}\ (\mu\in[m])$ as well as one more node $s$ ($s$ is removed when considering confluent flows). Meanwhile, for each triple $T_{\mu}$ in $T$, there are two vertices $x_{\mu}$, $y_{\mu}$ to represent it. We connect $s_i$ with $s$, and $s$ with $b_i$ for each $i\in[n]$; we also connect $s'_{\mu}$ with $s$, and $s$ with $x_{\mu}$ for each $\mu\in[m]$. Similarly, we connect $t$ with $a_{il}$, $c_i$ for each $i\in[n]$ and $l\in[M-1]$. For each tuple, $T_{\mu}=(a_{p_{\mu}},b_{q_{\mu}},c_{r_{\mu}})$, we connect  $x_{\mu}$ with $b_{q_{\mu}}$, and $y_{\mu}$ with $c_{r_{\mu}}$ as well as $(M-1)$ copies of $a_{p_{\mu}}$.

{\em Edge capacities and lengths.} All edge capacities are set as 1. See Figure~\ref{fig:bicriteria} for details. Let each $(s'_{\mu},x_{\mu})$ have length 5 (red edges), and each $(y_{\mu},c_{r_{\mu}})$ have length 4 (green edges), and each $(a_{il},t)$ have length 3 (blue edges), and all other edges have length  2 (black edges). Finally, we set the time horizon $T=14$ in the constructed graphs for the  Unsplittable/Confluent Maximum Flow Over Time problems.
Based on the constructed instance, we have (proof in Appendix \ref{App: Proofs of Constant Bicriteria Approximation Hardness})
\begin{theorem}\label{thm-biapproxhardQuick}
There exists a constant $\alpha>0$ such that, for any $\epsilon>0$, there is no polynomial-time  $(1+\alpha,\frac{15}{14}-\epsilon)$-approximation for the Unsplittable/Confluent Maximum Flow Over Time problem in both directed and undirected graphs, unless $P = NP$.
\end{theorem}
To show the hardness of the Unsplittable/Confluent Quickest Flow problem, we
construct an instance similar to Theorem~\ref{thm-biapproxhardQuick}, except that we let each source have supply 1, and have
\begin{theorem}\label{thm-biapproxhardTime}
There exists a constant $\alpha>0$ such that, for any $\epsilon>0$, there is no polynomial-time  $(\frac{15}{14}-\epsilon,1+\alpha)$-approximation for the Unsplittable/Confluent Quickest Flow problem in both directed and undirected graphs, unless $P = NP$.
\end{theorem}



\section{Polylogarithmic Approximation for Confluent Dynamic Flows}
\label{Sec:Polylog_Approx_Confluent}

\vspace{-0.3cm}
\subsection{Static Confluent Flows in Uncapacitated Networks with $\kappa$ Sources}\label{sec:uncapacitated_network}
We now  develop techniques for routing confluent flows in uncapacitated networks with $\kappa\leq n$ sources.
Through Section~\ref{sec:StaticConFlow}, unless otherwise specified, the flow discussed is {\em static.} 

\vspace{-0.2cm}\begin{definition}[$\beta$-Satisfiable]
For any $\beta\in [0,1]$, a supply $d_i$ is  {\em  $\beta$-satisfiable} in  flow $f$ if at least a $\beta$ faction of $d_i$ can be sent to the sink via $f$.
A flow $f$ is $\beta$-satisfiable  if all supplies are $\beta$-satisfiable in $f$.
\end{definition}\vspace{-0.2cm}

Again, suppose $G = (V,A)$ is  a static directed graph with supply $d(v)$ located at each $v\in V$. There exists a collection of sinks $\{t_1,...,t_k\}\subset V$. We let $\kappa$ be the number of non-zero supplies, and let all edge  and node capacities be 1. We present (proof in Appendix~\ref{apd:uncapacitated_network})
\vspace{-0.0cm}\begin{theorem}\label{thm-mergetree}
In the directed uncapacitated  network with $\kappa$ uniform non-zero supplies, given a (splittable) $1$-satisfiable flow $f$,  there exists a randomized algorithm for finding  a multi-sink confluent flow $f'$ with the node congestion bounded by $O((NC(f))^2\log^3 \kappa)$ whp\footnote{Throughout the paper, we use {\em whp} to  mean with high probability, or more precisely, with probability $1-n^{-c}$, where $n$ is the number of nodes in the network and $c$ is a constant.}.
\end{theorem}
Note that if $\kappa$ is bounded and $f$ is feasible, Theorem~\ref{thm-mergetree} can provide confluent flows with constant congestion.



Also, the support of the resulting flow is a collection of trees rooting at those sinks $t_1,...,t_k$. We guarantee  that the height of those trees can be bounded as below.
\vspace{-0.2cm}
\begin{lemma}\label{lmm-mergetree}
Whp, the height of any tree constructed in the randomized algorithm is at most $O(NC(f)\log n)$.
\end{lemma}

\vspace{-0.4cm}
\subsection{Static Length-Bounded Confluent Flows in Monotonic Networks}\label{sec:monotonic_network}

This section gives an algorithm for constructing a length-bounded confluent flow in {\em monotonic networks}, utilizing techniques developed in Section~\ref{sec:uncapacitated_network}.
A monotonic network is a special (static) directed graph with  vertex capacities and no edges pointing in the direction of decreasing capacity.  Formally,
\vspace{-0.2cm}
\begin{definition}[Monotonic Network]
A directed graph $G=(V,A)$ with node capacity $c(v)$ for each $v\in V$ is a  {\em monotonic network} iff $c(u)\leq c(v)$ for every arc $(u,v)$.
\end{definition}

\vspace{-0.2cm}
The network $G = (V,A)$ is the same as Section~\ref{sec:uncapacitated_network}
except that here each node has capacity $c(v)$ and  each edge has capacity 1.
Our first step is to prove (Appendix \ref{apd:monotonic_network})
\vspace{-0.2cm}
\begin{theorem}\label{thm-monotonicflow}
Let $G = (V,A)$ be a monotone network. Given a $1$-satisfiable flow $f$ with node congestion at most 1, one can, in polynomial time, construct a confluent $1$-satisfiable flow with node congestion $O(\log^8 n)$ and flow length $O\left( L(f)\log n \log c_{\max}/ \log\log n \right)$ whp, even without the no-bottleneck assumption.\vspace{-0.2cm}
\end{theorem}
The idea is to first decompose the monotonic network into several sub-networks, and
in each, construct  length-bounded confluent flows with small node congestion. Connecting all confluent flows in those sub-networks, we can construct a confluent flow in the original network as desired.
Our monotonic network technique incorporates a new parameter, namely the edge length, and, more importantly, our objective is to construct  a {\em bicriteria} confluent flow, namely bounding {\em both}  node congestion and length (note that  in \cite{shepherd2015polylogarithmic}, only  node congestion can be bounded). The main difference from~\cite{shepherd2015polylogarithmic} lies in that we embed our new algorithms for uncapacitated networks into the monotonic network routing.

Our technique can be further improved if we remove the length-bounded constraint. The key observation is that the sources in each sub-network are only induced by the given (splittable) flow that we would like to re-route into a confluent one.  We can guarantee that, if the given flow is unsplittable, at most $\kappa$ flow paths pass between two sequential sub-networks, inducing at most $O(\kappa)$ sources. This, combined with our new technique for uncapacitated networks, gives the improvement of the congestion from poly$(\log n)$ to poly$(\log \kappa)$.
\vspace{-0.2cm}
\begin{theorem}\label{thm-monotonicflow-relaxed}
Let $G = (V,A)$ be a monotone network with a single sink. If there is $1$-satisfiable flow $f$ with node congestion at most 1, one can, in polynomial time, construct a confluent $1$-satisfiable flow with node congestion $O(\log^8 \kappa)$ whp, under the NBA.
\end{theorem}

\vspace{-0.4cm}
\subsection{Static Length-Bounded Confluent Flows in General Networks}\label{sec:StaticConFlow}

Via the techniques developed above for monotonic networks, this section develops a polynomial-time algorithm for determining a length-bounded confluent static flow in general networks. 

Suppose we are given a directed/undirected {\em edge-capacitated} network $G(V,E)$ 
(Section ~\ref{sec:monotonic_network} dealt with {\em node} capicitated networks).
Each node $v\in V$ has a supply $d(v)$ to be sent to the unique sink $t$. Our goal is to find {\em a subset of supplies} of maximum total value that can be routed via a confluent flow, whose flow length and  edge congestion are both bounded.


To this end, we need to pre-process the network as follows.
First,
we ignore those demands of size at most $d_{\max}/2\kappa$, as they contribute at most half of the value of the optimal flow. Meanwhile, we round each supply up to the nearest power of 2, and group those with the same value together, producing $O(\log \kappa)$ groups of distinct supply sizes.
To compute an approximation, we will separately route each supply group in $G$, and output the flow of the maximum value among all groups. Note that, this will lose a $O(\log \kappa)$ factor in the approximation ratio.
Hence, we reduce the original problem to the uniform-supply case. Without loss of generality, by scaling, we can assume every supply is 1.

Second, we round each capacity up to the nearest power of 2, and assume all edges have capacity at most $\kappa d_{\max}$, i.e., $c_{\max}\leq \kappa d_{\max}$ as the extra capacity above this value is superfluous.
Furthermore, when considering the uniform-supply case, those edges with capacity less than the supply size would never be used, as the supply should be routed confluently. Accordingly, we can assume each edge capacity is in $[1,\kappa]$ as $d_{\max}=1$
in unit-supply case, and then there exist $O(\log \kappa)$ distinct capacity sizes.

Given a directed/undirected edge-capacitated network $G(V,A)$ with a single sink $t$, letting $k:=\lfloor \log c_{\max}\rfloor + 1$, we construct the directed {\em $k$-layer (monotonic) network}  $H$ (see Figure~\ref{fig:3layer}):
\begin{enumerate}

\vspace{-0.cm}

\item[$\bullet$] {\bf \em  $k$ layers.} Create $k$ layers and $k$ node sets $V(H_{0}),V(H_{1}),...,V(H_{k-1})$, where $V(H_{i}):=V(G)\setminus \{t\}$ and the $i$-th layer contains $V(H_{i})$.

\vspace{-0.cm}

\item[$\bullet$] { \bf\em Induced node capacities.} For the $i$-th node set $V(H_{i})$ $(i=0,...,k-1)$, denote by $u^{i}$ the $i$-th copy of node $u$, and let $u^i$ have capacity $2^i$.

    \vspace{-0.cm}

\item[$\bullet$] {\bf \em Vertical arcs.} For each edge $(u,v)\in A(G)$, connect two vertical arcs $(u^i,v^i)$ (and $(v^i,u^i)$ if $G$ is undirected) with capacity of $2^i$ in $H$, iff the capacity of $(u,v)$ is at least $2^i$ $(i=0,...,k-1)$.

    \vspace{-0.cm}

\item[$\bullet$] {\bf \em Horizontal arcs.} For $0\leq i\leq k-2$, $\forall u\in V$, connect a horizontal arc $(u^i,u^{i+1})$ with capacity $2^i$.

    \vspace{-0.cm}


    \vspace{-0.cm}

\item[$\bullet$] {\bf \em Arc lengths.} Let vertical arcs have the same length as arcs in $G$, and horizontal arcs have length 0.

    \vspace{-0.cm}

\item[$\bullet$] {\bf  $H:=(V(H),A(H))$.} Set $V(H)$ as the union of $V(H_{0}),V(H_{1}),..., V(H_{k-1}), \{t\}$ plus those dummy sinks, and set $A(H)$ as the collection of those vertical and horizontal arcs.

    \vspace{-0.cm}

\item[$\bullet$] {\bf \em Supplies.} Place the supply of $v$ at its copy $v^0$ in Layer 0.
\vspace{-0.cm}

\item[$\bullet$] {\bf \em Dummy sinks.} If there exists an edge $(u,t)$ with capacity of $2^i$, then create a copy $t_u^{j}$ of $t$ in  Layer $j$ and let the capacity of $t_u^{j}$ be $2^j$, for each $j=i,...,k-1$. Connect the vertical arc $(u,t_u^{i})$ with capacity of $2^i$, and the horizontal arc $(t_u^{j},t_u^{j+1})$ with capacity of $2^j$, for each $j=i,...,k-2$. Finally, connect the arc $(t_u^{k-1},t)$ with capacity of $2^{k-1}$.

\end{enumerate}

Our multi-layer network can be viewed as a new construction enabling our length-bounded routing technique to work in edge-capacitated networks.
Applying Theorem~\ref{thm-monotonicflow} yields:

\vspace{-0.2cm}\begin{theorem}\label{thm-klayerflow}
In the layered network $H$, given a (splittable) flow $f$ for routing all unit supplies with node congestion at most 1, there exists a  polynomial-time algorithm for constructing  a $1$-satisfiable confluent flow with node congestion $O(\log^8 n)$ and flow length $O\left( L\log^2 n/\log\log n \right)$ whp.
\end{theorem}


Thus, via Theorem~\ref{thm-klayerflow},  we can obtain a confluent flow $h$ in the $k$-layer network $H$ with both node congestion and length being  bounded. Nevertheless, since $H$ is constructed from logarithmic copies of nodes in $G$, the constructed  confluent flow $h$ in $H$ may induce a non-confluent flow in $G$, because some vertices $v$ might contain logarithmic out-flow edges. We show by Lemma~\ref{lmm-rerouteflow} that there is a polynomial-time scheme for re-routing $h$ into a confluent flow in the original network $G$.
Also, although  we bound  {\em node} congestion in $H$, the original network $G$ is in fact {\em edge}-capacitated and we are actually interested in the edge congestion. Fortunately, our construction of multi-layer networks can be patched. With the help of the monotonic structure and dummy sinks, we can bound the edge congestion by Lemma~\ref{lmm-edgecongestion}.

Combining everything,  we conclude that (see also Appendix~\ref{Apd:StaticConFlow})
\vspace{-0.2cm}\begin{theorem}\label{thm-generallengthboundflow}
Suppose $G$ is a directed/undirected edge-capacitated network with one sink. If there is an $L$-length-bounded confluent flow for routing all supplies with edge congestion at most 1 in $G$, then, there exists a polynomial-time algorithm for finding a confluent flow for routing {\em a subset of supplies} with value
at least $\sum_{i\in[\kappa]}d_i/O(\log^{2} \kappa)$, with edge congestion $O(\log^{8} n)$ and  flow length  $O(L \cdot \log^3 n/\log\log n)$ whp.
\end{theorem}

\vspace{-0.4cm}
\subsection{Polylogarithmic Approximation for the Confluent Dynamic Flows}
\vspace{-0.2cm}

With the techniques developed and  transformations
(Lemmas~\ref{lmm-trans1} and~\ref{lmm-trans2}),  the polylogarithmic approximation for the confluent dynamic problem  immediately follows. Note that our algorithms  
do not  use any storage at intermediate nodes.

\vspace{-0.2cm}\begin{theorem}\label{thm-PolylogConQuickFlow}
In directed/undirected, edge-capacitated dynamic networks, there is a polynomial-time algorithm that constructs an $(O(\log^{8} n), O(\log^{2}\kappa))$-approximation  for the single-sink Confluent Quickest Flow problem whp.
\end{theorem}
\vspace{-0.2cm}
\begin{theorem}\label{thm-PolylogConFlowOverTime}
In directed/undirected, edge-capacitated dynamic networks, there is a polynomial-time algorithm that constructs an $(O(\log^{2} \kappa), O(\log^{8}n))$-approximation for the single-sink Confluent Maximum Flow Over Time problem whp.
\end{theorem}
Our technique can be  restricted to static flows, yielding 
\vspace{-0.2cm}\begin{theorem}\label{thm-StaticConfluent}
In directed/undirected, edge-/node-capacitated static networks that satisfy the no-bottleneck assumption, there is a polynomial-time algorithm that constructs an $O(\log^{10}\kappa)$-approximation for the single-sink Demand Maximization Confluent Flow problem whp.
\end{theorem}



{{\bf Acknowledgement:} 

We would like to thank the authors of \cite{shepherd2015polylogarithmic} for providing us with  a pre-print of the full version of their paper}

\vspace*{-.2 in}

{
\bibliography{Evacuation_problems}
}

\newpage

\appendix

\section{An Example of a Dynamic Flow Problem}\label{Apd:Preliminaries}

Examples of quickest flows in both the unsplittable and confluent cases are given in Figure \ref{fig:Types of Flow} so as to provide intuition.

\begin{figure}
\begin{center}
\includegraphics[scale=0.7]{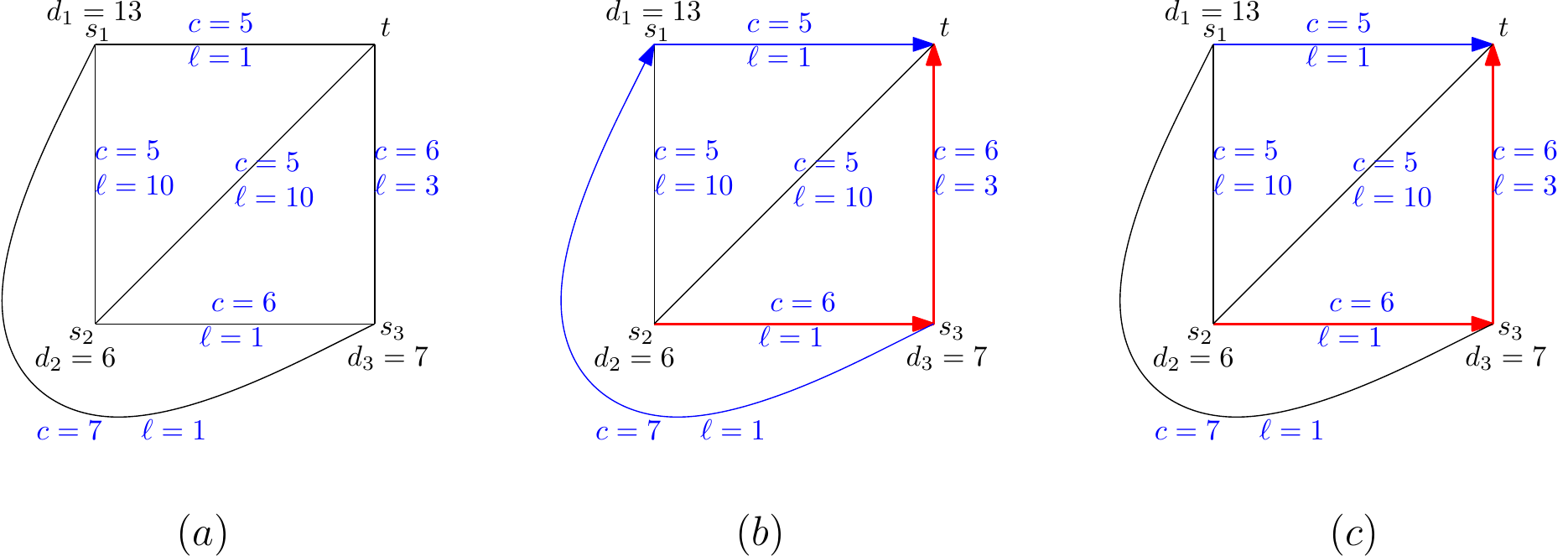}
\caption{(a) gives a simple example of a dynamic flow graph.
(b) and (c) illustrate quickest flows for this input that move the supplies $d_i$   from $s_1,s_2,s_3$ to $t.$ (b) is an optimal quickest for the unsplittable flow condition.
It routes $s_1 \rightarrow t,$  $s_2\rightarrow s_3 \rightarrow t$ and $s_3 \rightarrow s_1 \rightarrow t$.  In this solution, the flow from $s_3$ orignally needs  to wait at $s_1$ to get onto the edge $(s_1,t)$ but all flow from $s_3$ finishes arriving at $t$ at 4 time units, the same time as the flow from $s_2$ fully arrives at $t$.  Note that (b) is {\em not} a confluent flow since some flow at $s_3$  (the original supplies there),  leave via $(s_3,s_1)$ while other flow from $s_3$ (the flow starting at $s_2$), leaves via $(s_3,t)$.
(c) is an optimal quickest flow for the confluent flow condition.  It routes   $s_1 \rightarrow t$  and  $s_2\rightarrow s_3 \rightarrow t$.
Note that, because of confluence,  $s_2\rightarrow s_3 \rightarrow t$ implies $s_3 \rightarrow t$, i.e., all of the flow {\em starting} at $s_3$ must also leave through the edge $(s_3,t).$ All flow reaches $t$ after $5$ units of time.
}\vspace{-0.8cm}
\label{fig:Types of Flow}
\end{center}
\end{figure}


\section{Proofs of Constant Approximation Hardness of Single-Sink Unsplittable and Confluent Quickest Flows}\label{apd:Proof_of_Apd_2_Approx}


We start by constructing the  reduction from  Two-Disjoint Paths (Capacitated) to Quickest Flow.
We consider both directed and undirected graphs but only show the  details for undirected graphs (since directed graphs are similar).
Let $\mathcal{I}$ be an instance of the undirected version of the Two-Disjoint Paths (Capacitated) problem.   
We construct a network  based on $\mathcal{I}$, by connecting the unique sink $t$ with $y_1$ and $y_2$, and connecting  two sources $s_1$ and $s_2$ with $x_1$ and $x_2$, respectively. Let all edge lengths be 1. Choose the parameter $\beta=2\alpha$ in $\mathcal{I}$, and set the capacities of $(s_1,x_1)$ and $(y_1,t)$ as $\alpha$, and those of $(s_2,x_2)$ and $(y_2,t)$ as $\beta$ (see Figure~\ref{fig:3/2instance}). Finally, place supplies $M\alpha$ and $M\beta$ on the sources $s_1$ and $s_2$, respectively, where $M$ is a sufficiently large constant, $p$ is the number of vertices in $\mathcal{I}$ and  $M\gg p$.


Based on the instance constructed, we prove Theorem~\ref{thm-2approfixed} as follows.
\begin{proof}[Proof of Theorem~\ref{thm-2approfixed}]
We first consider the proof for unsplittable flows.
In the undirected graphs,
if $\mathcal{I}$ is a YES-instance (see Figure~\ref{fig:3/2instance}(a)), then there exist two edge-disjoint paths connecting $t$ with $s_1$ and $s_2$, respectively such that the path connecting $s_2$ with $t$ uses only edges with capacity $\beta.$  Recall that by construction the first edge on the path from $s_1$ has capacity $\alpha.$
Thus,  the optimal plan is to send $\alpha$ (and $\beta$) units of supplies along the edge-disjoint path from $s_1$ (and $s_2$) to $t$. The  time for routing all supplies  is then at most $M+p+2$. Note that the flow constructed is unsplittable.


However, if $\mathcal{I}$ is a NO-instance, then either there don't exist two edge-disjoint paths connecting $t$ with $s_1$ and $s_2$ (see Figure~\ref{fig:3/2instance}(b)), or there exist two edge-disjoint paths but the path from $s_2$ to $t$ must use  at least one edge with capacity $\alpha$ (see Figure~\ref{fig:3/2instance}(c)). In the first case at most $\beta$ units of flow can reach $t$ in any time unit and in the second case at most $2 \alpha$ units of flow can reach $t$ in any time unit.
Thus, since  a total of  $M(\alpha+\beta)$ units of flow need to be routed, the time for routing all supplies to $t$ via an unsplittable flow is at least $M(\alpha+\beta)/(2\alpha)=3M/2$ (this is a rough lower bound since it doesn't take into account the time that it takes for the first unit of flow to arrive at $t$).

Let $\epsilon:=\frac{3(p+2)}{2(M+p+2)}>0$. It immediately follows that, if the single-sink Unsplittable Quickest Flow problem in undirected graphs
can be approximated to within a factor $3/2-\epsilon$, then one can determine whether the instance $\mathcal{I}$ is a YES- or NO-instance, which is NP-hard in undirected graphs. This immediately implies   the $3/2$-approximation hardness in the Fixed-Sink setting for unsplittable flow.

For the Confluent Quickest Flow in undirected graphs, the instance $\mathcal{I}$ to be used is the node-disjoint version of Two-Disjoint Paths (Capacitated) problem. The hardness result follows the same analysis.

The analysis of directed graphs is similar, except that we use the directed version of Two-Disjoint Paths (Capacitated) problem, and connect $s_i$ with $x_i$ by an arc $(s_i, x_i)$, and connect $y_i$ with $t$ by an arc $(y_i, t)$ ($i=1,2$). This completes the proof.
\end{proof}

\begin{figure}[t]
\begin{center}
\vspace{-.0cm}\hspace{-0.3cm}{\includegraphics[scale=0.4]{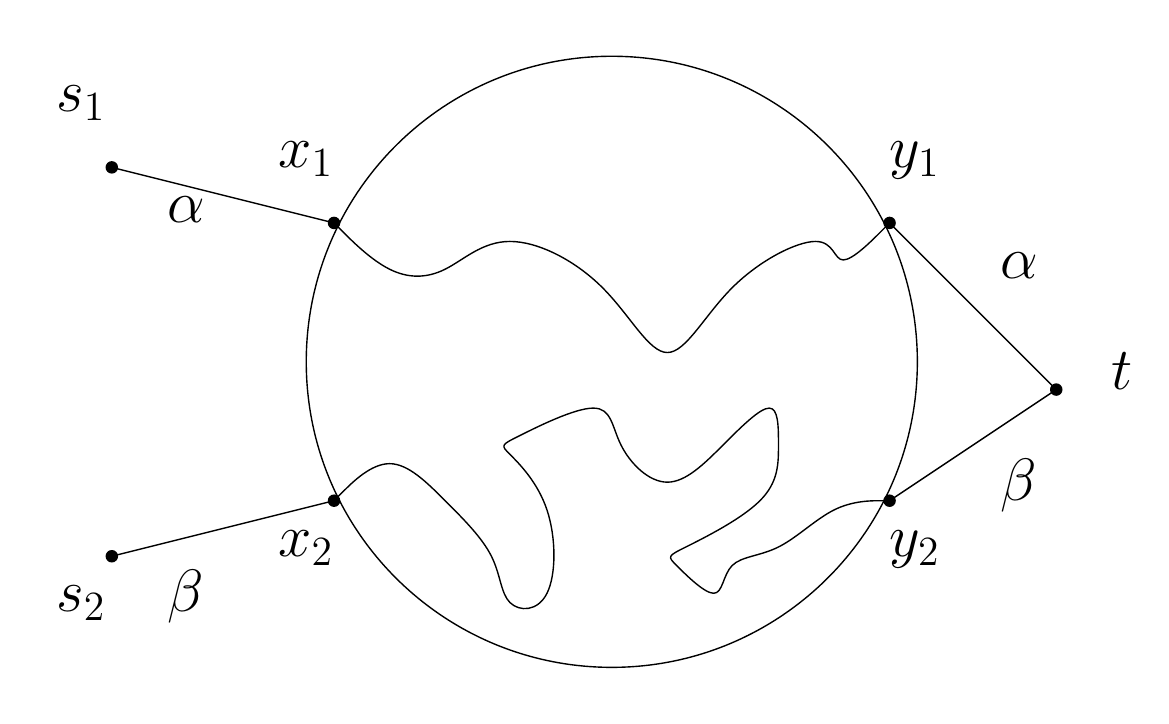}}\hspace{-0.3cm}
{\includegraphics[scale=0.4]{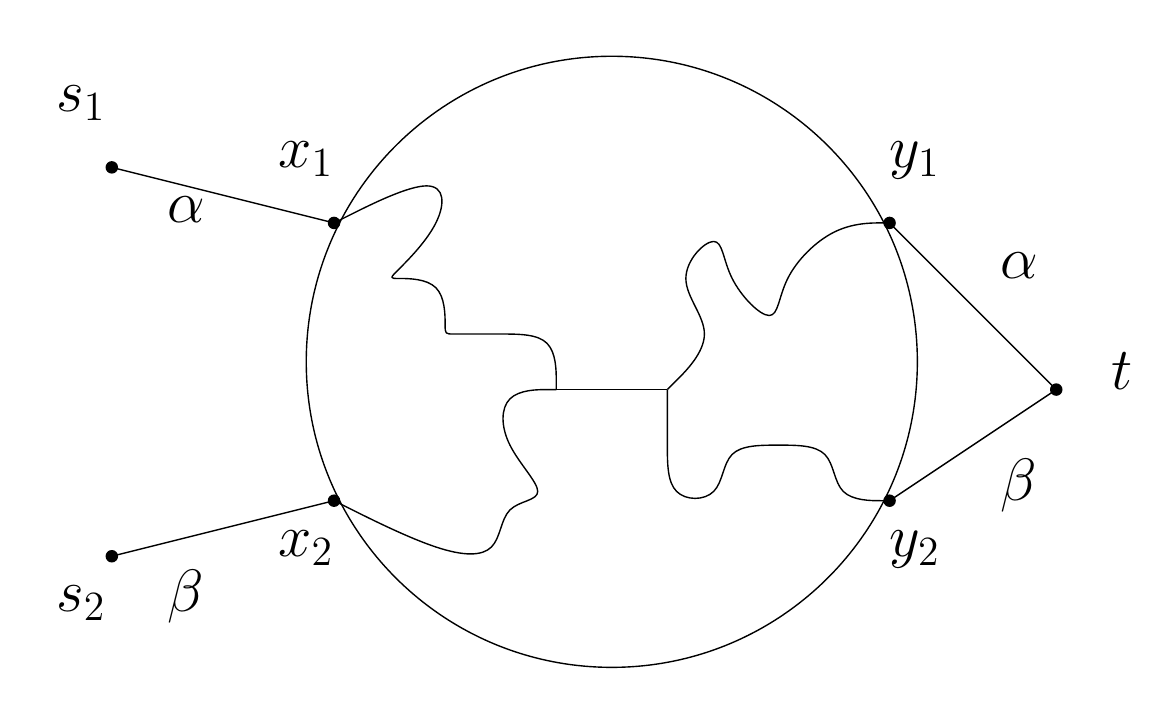}}\hspace{0.2cm}\hspace{-0.3cm}
{\includegraphics[scale=0.4]{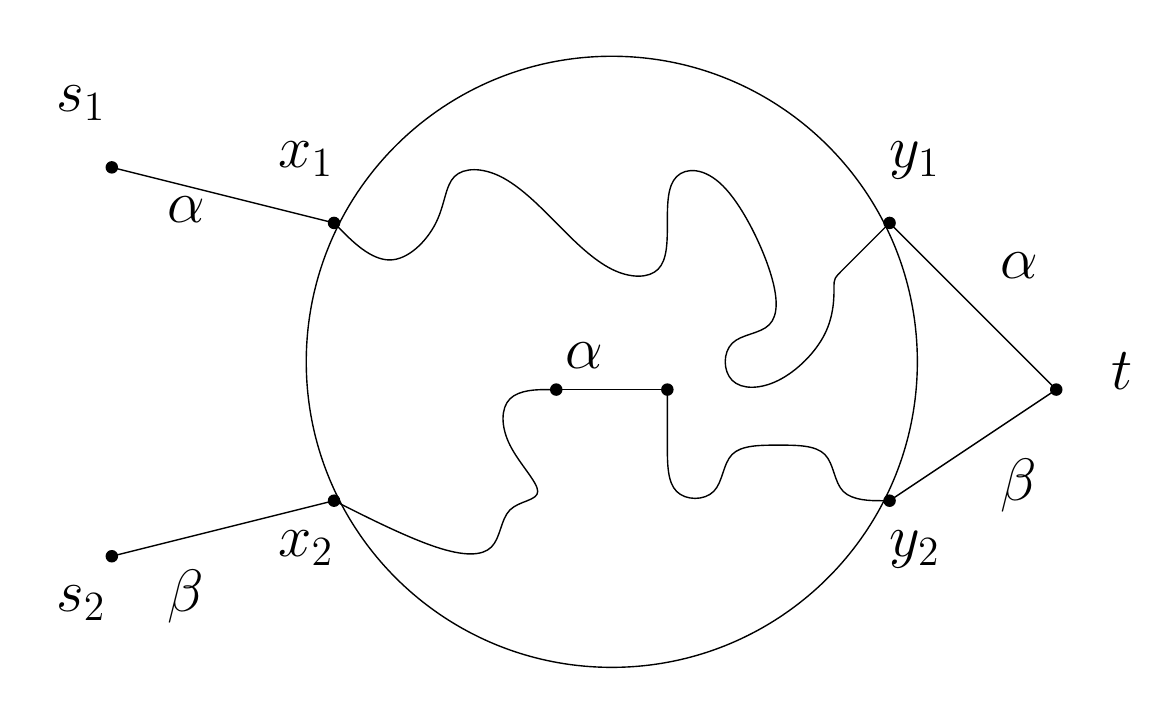}}\hspace{0.2cm}$ $
\\ \vspace{0cm}{\footnotesize\hspace{0cm}(a) Yes-Instance\hspace{2cm}(b) No-Instance case (i) \hspace{2cm} (c) No-Instance case (ii)}
\\ \caption{
Hard instances for the  single-sink Unsplittable Quickest Flow 
  and Maximum  Flow Over Time 
  problems.
}\vspace{-0.8cm}
\label{fig:3/2instance}
\end{center}
\end{figure}



\section{Proofs of Logarithmic Approximation Hardness of Single-Sink Confluent Quickest Flows}\label{App:ApdLogHardness}



\begin{figure}
\begin{center}
\vspace{-0cm}\hspace{-0.4cm}{\includegraphics[scale=0.4]{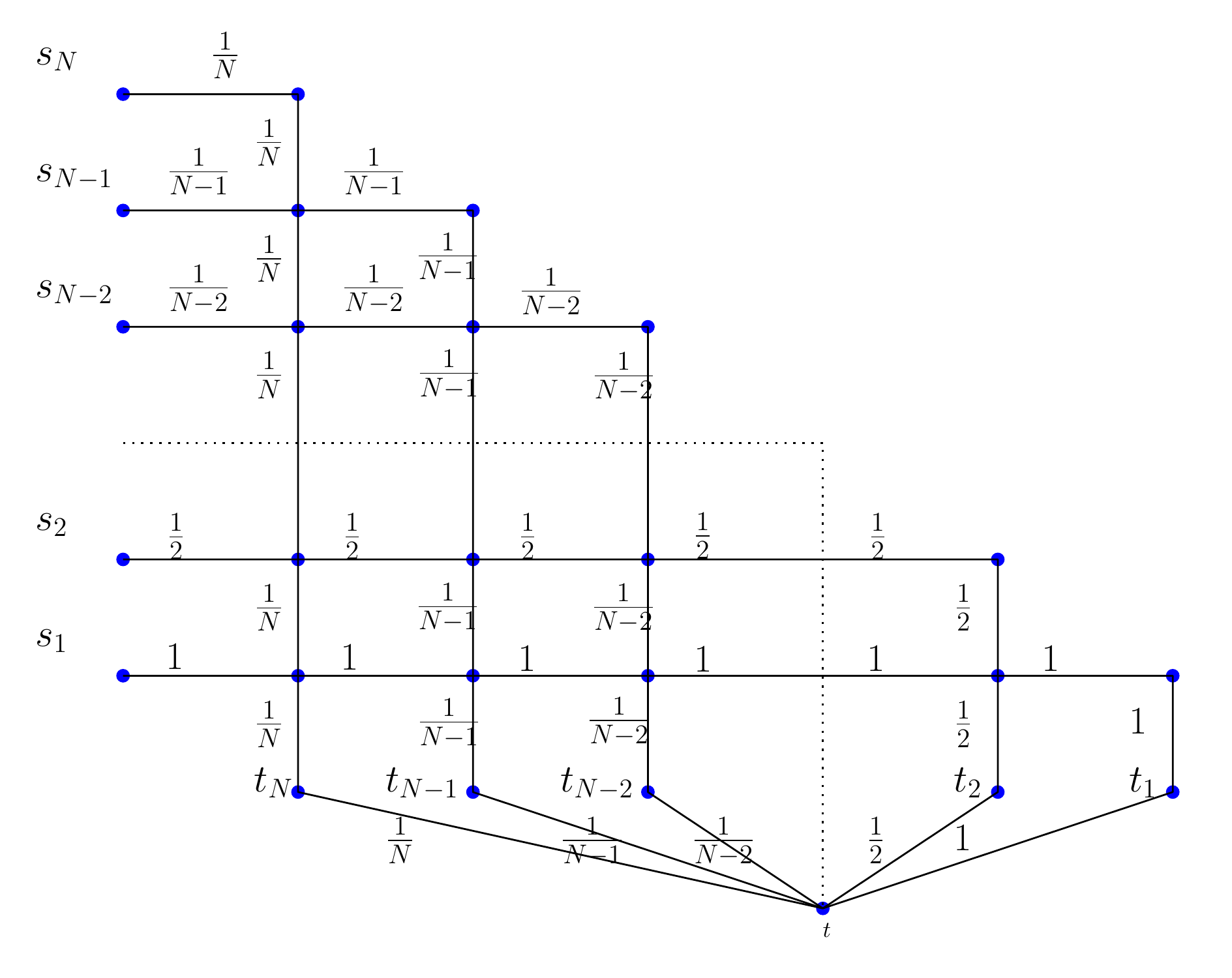}}
\\ \caption{Hard instance for the  single-sink
confluent dynamic flow problem.}\vspace{-0.8cm}
\label{fig:half-grid}
\end{center}
\end{figure}

Based on the hard instance in Figure~\ref{fig:half-grid}, we have the following lemmas.
\begin{lemma}\label{lmm-singlenode1}
If $\mathcal{I}$ is a YES-instance, then it takes at most $M^2+Np+N+2$ time for all supplies to be sent to the sink $t$ via a dynamic confluent flow.
\end{lemma}
\begin{proof}
If $\mathcal{I}$ is a YES-instance, there exist two node-disjoint paths inside $\mathcal{I}$: One is from $x_1$ to $y_1$, and the other is from $x_2$ to $y_2$. Hence,  the  source $s_i$ can send $1/i$ supplies per unit time, without affecting others, along the $i$-th row and then down along the $i$-th column to $t_i$ (i.e., the $i$-th {\em canonical path}).
It then takes at most $M^2+Np+N+2$ time to route all the supplies at $s_i$ to $t$, since the length of each canonical path inside $\mathcal{G}$ is at most $Np+N+2$, where $p$ is the number of vertices inside $\mathcal{I}$. The proof is complete.
\end{proof}

Now we consider the case when $\mathcal{I}$ is a NO-instance. For a confluent dynamic flow $f$, its support in $G_N$ is a collection of trees rooting at certain $t_i$'s. Denote the set of those trees as $\mathcal{T}=\{T_{i_1},...,T_{i_k}\}$, where $T_{i_j}\ (j\in[k])$  stands for the union of all the paths carrying the flow and terminating at $t_{i_j}$. Without loss of generality, suppose $i_1>i_2>...>i_k$.

Clearly, those trees in $\mathcal{T}$ are {\em edge-confluent} in $G_N$. That is, when two flow paths share an edge, they must  follow the same path to the same terminal $t_i$. Indeed, if two flow paths share an edge in $G_N$, they must, because the flow is a confluent flow in $\mathcal{G}$,  merge and go to the same terminal.  However, because each degree-4 node in $G_N$ is not a usual node -- it just signifies the  embedding of an   instance $\mathcal{I}$ there --   two flow paths sharing a node in $G_N$ might   go along two node-disjoint paths {\em inside} $\mathcal{I}$ and then continue on  to different terminals. Recall that a degree-4 node in the $i$-th column and $j$-th row of $G_n$ has two  incident vertical edges and two incident horizontal edges.  If two flow paths go along two node-disjoint paths inside  $\mathcal{I}$ and one of the flow paths both enters and leaves  along the vertical  edges and the other both enters and leaves along the horizontal  edges we say that  the two flows are ``crossing'' in $G_N$. (If one enters via a horizontal and leaves via a vertical and the other enters via a vertical and leaves via a horizontal this is not considered ``crossing''.)  Finally,   we say that two {\em trees}  in $\mathcal{T}$ are crossing in $G_N$ if
there is a pair of flows, one from each tree, that are crossing in $G_N$.

The crucial observation is that,  because the instance $\mathcal{I}$ is a NO instance of the directed version of the Two-Disjoint Paths problem,
the  trees in $\mathcal{T}$  must be  non-crossing in $G_N$. That is, for any $i, j\in[k]$ satisfying $i\neq j$, $T_i$ doesn't cross $T_j$. This is because the fact that it is a NO instance means that  there don't exist two node-disjoint paths inside $\mathcal{I}$ such that one is from to $x_1$ to $y_1$ and the other is from $x_2$ to $y_2$.
Under this circumstance, the crossing of $T_i$ and $T_j$ in $G_N$ must induce the merging of two flows in $T_i$ and $T_j$ inside $\mathcal{I}$ to use the same edge, resulting in those flows continuing on to the same terminal, contradicting  the fact that $t_i\neq t_j$.

To bound the maximum flow that can be sent to $t$ via $f$ per unit time, we would like to find  a cut in $\mathcal{T}$ with bounded  (capacity) weight.
Indeed, for those non-crossing, edge-confluent trees in $\mathcal{T}$, we have the following lemma.
\begin{lemma}\label{lmm-treecut}
For a set of non-crossing, edge-confluent trees in $G_N$, i.e., $\mathcal{T}=\{\mathcal{T}_{i_1},\mathcal{T}_{i_2},...,\mathcal{T}_{i_{k^*}}\}$, there exists a cut separating all sources in $\mathcal{T}$ from $t$, whose weight is at most $2$.
\end{lemma}
\begin{proof}
In the set of non-crossing, edge-confluent trees $\mathcal{T}=\{\mathcal{T}_{i_1},\mathcal{T}_{i_2},...,\mathcal{T}_{i_{k^*}}\}$,
we will recursively construct a set of  tree edges of $\mathcal{T}$ that  form a cut of $\mathcal{T}$,  separating all sources in $\mathcal{T}$ from $t$, with the cut weight at most $2$. Note that our technique here is an adaptation and generalization of  a  similar method in~\cite{NSV10,SV15}, which was used to prove the approximation hardness of the maximum throughput of the static confluent flow under the no-bottleneck assumption.

First, we give some notations. For $i \leq j$,  we define a subgrid $G(i, j)$ of $G_N$ induced by columns and rows whose indices lie in the range $[i, j]$. Let $r_1=1$, $l_1=N$ and $n_1=i_1$. Then, it is clear that all $\mathcal{T}_{i_j}$'s are located inside the grid $G(r_1,l_1)=G(1,N)$. Consider all paths of $\mathcal{T}_{i_1}$ from sources to $t_{i_1}$. Let $P_{i_1}$ be the highest path in $\mathcal{T}_{i_1}$, and $r_1'$ be the highest row number where the path $P_{i_1}$ intersects the $n_1$-th column, where the bottom node $t_{i_1}$ is located. Then, we define $r_2=r_1'+1$ and $l_2=n_1-1$ (recall that the row number increases from the bottom to the top, while the column number increases from the right to the left). Recursively, we can define $r_j'$, and then $r_j:=r_j'+1$ and $l_j:=n_{j-1}-1$ for $j=1,...,k$. Also, let $\mathcal{T}_{n_j}$ (the tree rooted at the bottom node $t_{n_j}$) be the leftmost tree passing through the subgrid $G(r_j,l_j)$.

Let the edge set $C:=\emptyset$ initially. For the paths of $\mathcal{T}_{i_1}$ in $G(r_1,l_1)$, since those paths finally goes to the bottom node $t_{i_1}$, we choose into $C$ the edge $e_{i_1}$ connecting $t_{i_1}$ with the bottom row. Note that $e_{i_1}$ separate the source nodes in $T_{i_1}$ from the bottom node $t_{i_1}$, and its capacity is $\frac{1}{n_1}$.
Now consider the paths of $\mathcal{T}_{i_2}$ in $G(r_1,l_1)$.  There exists three cases:
\begin{itemize}

\item [i.] All paths completely avoid routing through the subgrid $G(r_2,l_2)$ (see Figure~\ref{fig:three path cases}(a));

\item [ii.] Some paths go through $G(r_2,l_2)$, and some avoid $G(r_2,l_2)$ (see Figure~\ref{fig:three path cases}(b));

\item [iii.] All paths completely go through the subgrid $G(r_2,l_2)$ (see Figure~\ref{fig:three path cases}(c)).

\end{itemize}

\begin{figure}
\begin{center}
{\includegraphics[scale=0.4]{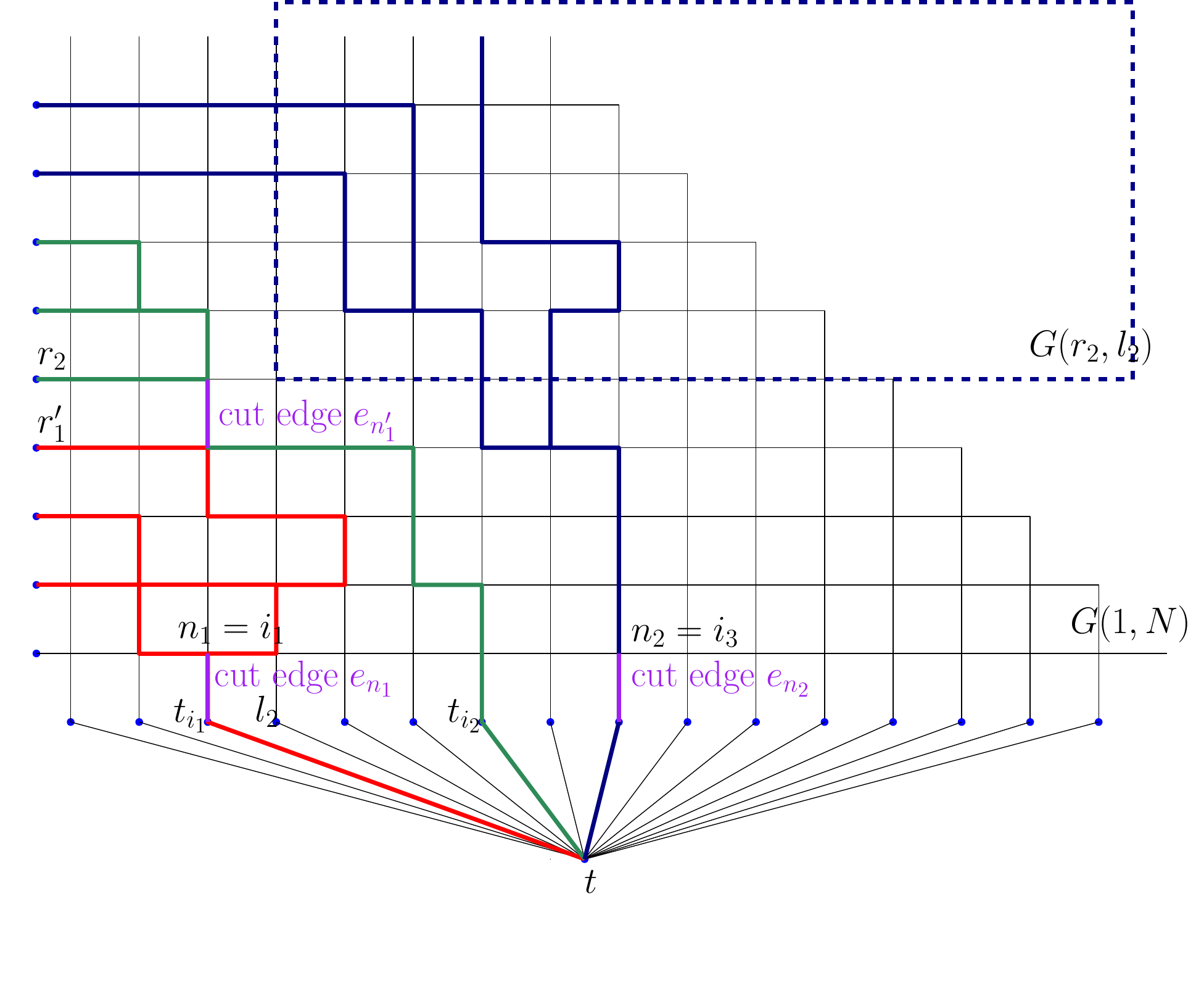}}\hspace{-0.1cm}
$ $
\\ \vspace{-0.5cm}{\footnotesize\hspace{0cm}(a) Case~i}
\end{center}
\begin{center}
{\includegraphics[scale=0.4]{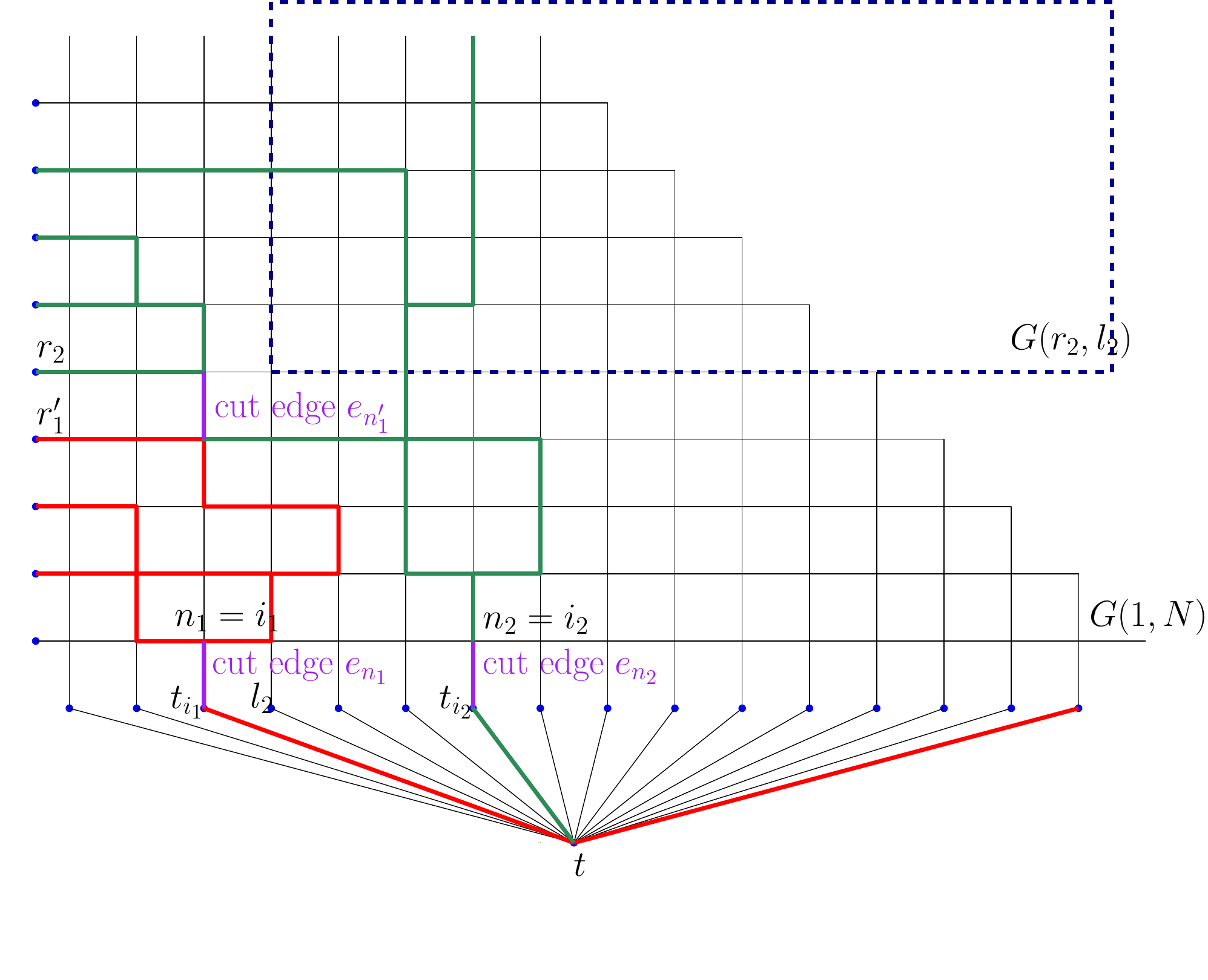}}\hspace{-0.1cm}$ $
\\ \vspace{-0.5cm}{\footnotesize\hspace{0cm}(b) Case~ii}
\end{center}
\begin{center}
{\includegraphics[scale=0.4]{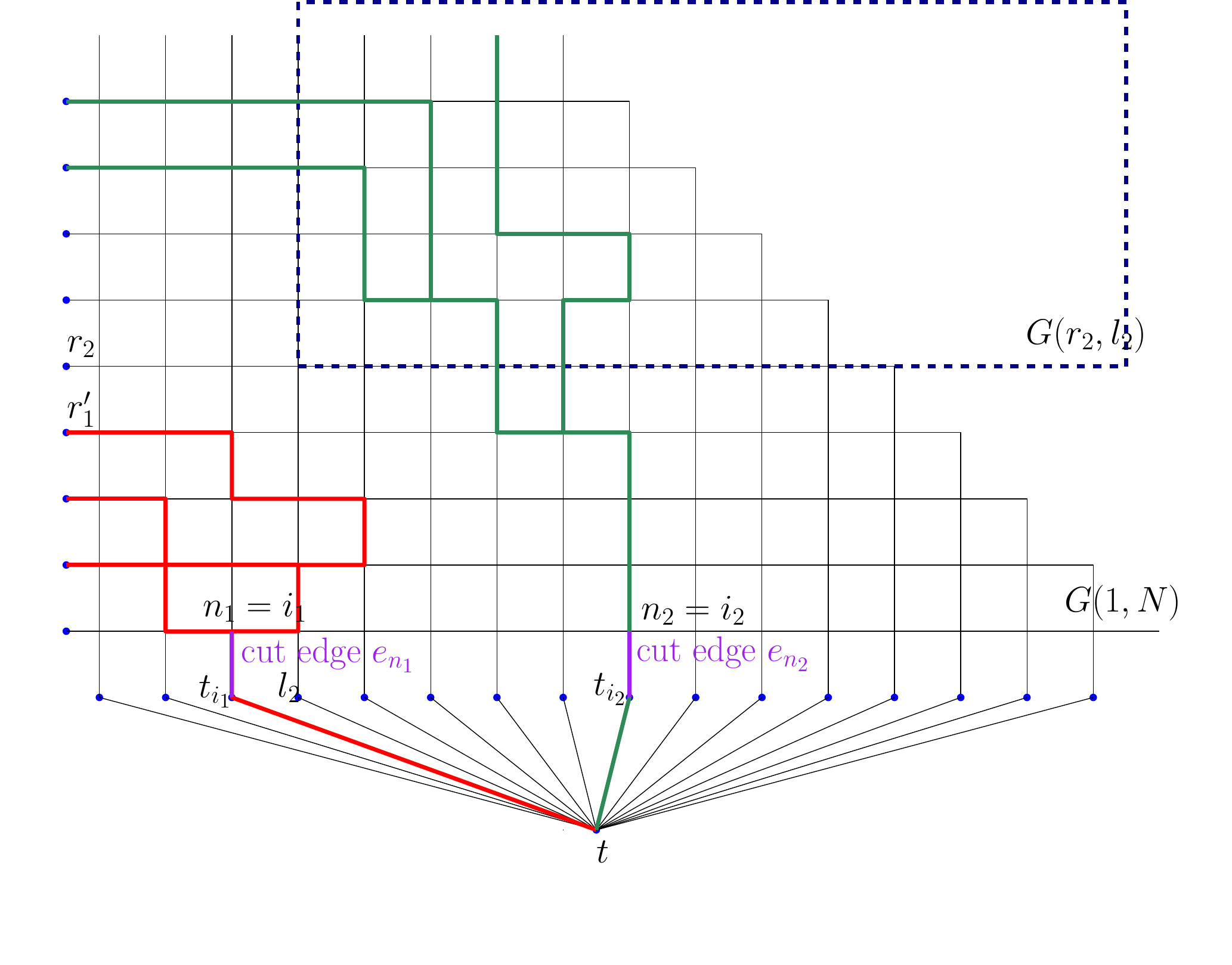}}\hspace{-0.1cm}$ $
\\ \vspace{-0.5cm}{\footnotesize(c) Case~iii}
\\ \caption{
Three cases of the paths of $T_{i_2}$ passing through $G(r_2,l_2)$. Red paths, green paths and blue paths represent $T_{i_1}$, $T_{i_2}$ and $T_{i_3}$, respectively. Purple edges denote the chosen cut edges in $C$.
}
\label{fig:three path cases}
\end{center}
\end{figure}

In Case~i, because $P_{i_1}$ is the highest path of $\mathcal{T}_{i_1}$ and the paths of $\mathcal{T}_{i_2}$ cannot cross the $P_{i_1}$, those paths avoiding $G(r_2,l_2)$ must go through the $i_1$-th column, and then there must be an edge in the $i_1$-th column carrying all those paths avoiding $G(r_2,l_2)$. We denote such an edge as $e_{i_2}$, and add it into $C$. In Case~ii, by a similar observation, it can be seen  that there must be an edge in the $i_1$-th column carrying all those paths avoiding $G(r_2,l_2)$. Again, we denote such an edge as $e_{i_2}$ and add it into $C$. Note that $e_{i_2}$ separates the terminal $t_{i_2}$ from the sources that are connected with those paths in $G(r_1,l_1)$, and its capacity is at most $\frac{1}{n_1}$. For those paths going through $G(r_2,l_2)$ in Case~ii or Case~iii, we would consider them in subgrid $G(r_2,l_2)$. Thus, it can be seen that, no matter which case happens, the cut edge we choose, i.e., $e_{i_1},e_{i_2}$, have total capacity at most $\frac{2}{n_1}$ (here we denote $n_1'=i_2$).

We repeat the above process to determine the cut edge in the subgrid $G(r_j,l_j)$ until $j=k$ or $r_j>l_j$. Consider those paths of $\mathcal{T}_{n_j}$ passing through $G(r_j,l_j)$. They must go through the edge connecting $t_{n_j}$ with the bottom row. We denote this edge as $e_{n_j}$, and add it into $C$. Note that this edge separates the sources from $t_{n_j}$, and its capacity is at most $\frac{1}{n_j}$. Let $n_j'$ be the index after $n_j$ in the set $\{i_1,i_2,...,i_k\}$. Then, similar to the analysis above, we know those paths of $\mathcal{T}_{n_j'}$ can completely or partly avoid routing through the subgrid $G(r_{j+1},l_{j+1})$. No matter which case happens, we choose the edge $e_{n_j'}$ such that $e_{n_j'}$ separates the bottom node $t_{n_j'}$ from the sources that are connected with the  paths of $\mathcal{T}_{n_j'}$ in $G(r_j,l_j)$, and its capacity is at most $\frac{1}{n_j}$. When the algorithm terminates, we have a set of cut edges $C$.

Since each subgrid contains at least one less tree in $\mathcal{T}$ than the subgrid before it, the number of iterations $k^*$ is less than $k$. Note that, for $j<k^*$, $\frac{1}{n_j}\leq\frac{1}{l_{j+1}}$, since $l_{j+1}=n_{j}-1$; for $j=k^*$, we have $\frac{1}{n_j}\leq \frac{1}{r_{k^*}}$, since $r_{k^*}\leq n_{k^*} \leq l_{k^*}$.
It can be seen that $l_1>...>l_{k^*}\geq r_{k^*}>...>r_1$, and $r_j\geq j$ for all $j$. This means that $r_{k^*}>k^*$ and $l_j\geq k^*$ for all $j$.
Now, we can bound the weight $w(C)$ of the induced cut as follows:
\begin{eqnarray}\label{ineq:treecut}
w(C) &\leq& \sum_{1\leq j\leq k^*} c(e_{n_j})+c(e_{n_j'})\leq \sum_{1\leq j\leq k^*} \frac{2}{n_j}\leq 2\left(\sum_{2\leq j\leq k^*} \frac{1}{l_j}\right)+\frac{2}{r_{k^*}} \leq 2  \sum_{1\leq j\leq k^*} \frac{1}{k^*}=2\nonumber\\
&&
\end{eqnarray}
This  completes the proof.
\end{proof}

\begin{lemma}\label{lmm-singlenode2}
If $\mathcal{I}$ is a NO-instance, then it takes at least $M^2H_N/2$ time to confluently route all supplies to $t$,
where $H_N:=1+\frac{1}{2}+...+\frac{1}{N}$.
\end{lemma}
\begin{proof}
Consider a confluent flow in $\mathcal{G}$. Then, the support of the flows in $G_N$ is a set of non-crossing, edge-confluent trees as stated before. Following the notation used before, denote the set of trees as $\mathcal{T}=\{T_{i_1},...,T_{i_k}\}$, where each $T_{i_j}$ is a tree rooted at the bottom node $t_{i_j}$. By Lemma~\ref{lmm-treecut}, we know there exists a cut that separates all sources in $\mathcal{T}$ from $t$, whose cut weight is at most $2$. Then, the amount of flow passing  through this cut is at most $2$ per unit time. Since the total weight of supplies is $\sum_{1\leq i\leq N}M^2\frac{1}{i}=M^2H_N$, it immediately means that the time for routing all supplies is at least $M^2H_N/2$.
\end{proof}

Utilizing Lemmas~\ref{lmm-singlenode1} and \ref{lmm-singlenode2}, we prove the logarithmic approximation hardness  for the Confluent Quickest Flow problem.
\begin{proof}[Proof of Theorem~\ref{thm-conlogappro}]
By Lemma~\ref{lmm-singlenode1}, we know if $\mathcal{I}$ is a YES-instance, then the time for routing all supplies to $t$ is at most $M^2+Np+N+2$. Here, we take a sufficiently large constant $M$, ensuring that the first term dominates the routing time. By Lemma~\ref{lmm-singlenode2}, we know if $\mathcal{I}$ is a NO-instance, then the time for routing all supplies to $t$ is at least $M^2H_N/2$.
It follows that if we could approximate the the routing time of the Confluent Quickest Flow problem in $\mathcal{G}$ to  a factor within  $H_N/2$, we could determine whether $\mathcal{I}$ is a YES- or NO-instance, which is NP-hard.

Note that $\mathcal{G}$ has $n = \Theta(pN^2)$ vertices, where $p$ is the number of vertices in $\mathcal{I}$. If we take $N = \Theta(p^{\frac{1}{2}(\frac{1}{\epsilon}-1)})$,
where $0<\epsilon < \frac{1}{2}$ is small, then $H_N = \Theta( \frac{1}{2}(\frac{1}{\epsilon}-1)\log p)$. Also, since $n=\Theta(p^{\frac{1}{\epsilon}})$, we have $H_N=\Theta( \frac{\epsilon}{2}(\frac{1}{\epsilon}-1)\log n)=\Theta( \log n )$.
%
Thus, it yields the bound as desired.
\end{proof}


\section{Proofs of Approximation Hardness of Single-Sink Maximum  Flow Over Time Problem}
\label{App: Apd MFT}

\subsection{Maximum Unsplittable  Flow Over Time}

This section gives the proof of the hardness of Unsplittable Maximum  Flow Over Time problem, i.e., Theorem~\ref{thm-unspovertime}.
\begin{proof}[Proof of Theorem~\ref{thm-unspovertime}]
We only need to consider the undirected case, since the directed case is similar except the embedded instance $\mathcal{I}$ is the directed version of Two-Disjoint Paths (Capacitated) problem, and edges are assigned with directions.  Set $T=M+p$, where $p$ is the number of vertices in $\mathcal{I}$ and $M$ is a large constant such that $M\gg p$.

If $\mathcal{I}$ is a YES-instance (see Figure~\ref{fig:3/2instance}(a)), then there exist two edge-disjoint paths connecting $t$ with $s_1$ and $s_2$, respectively, and hence the maximum value sent to $t$ within the time horizon $T$ is at least $(\alpha+\beta)(T-p)=3\alpha M$. However, if $\mathcal{I}$ is a NO-instance, then either there don't exist two edge-disjoint paths connecting $t$ with $s_1$ and $s_2$ (see Figure~\ref{fig:3/2instance}(b)), or there exist two edge-disjoint paths but the paths from $s_2$ to $t$ must use the edge with capacity $\alpha$ (see Figure~\ref{fig:3/2instance}(c)). Then, in either case, the  maximum value sent to $t$ within the time horizon $T$ is at most $2\alpha T=2\alpha (M+p)$. Thus, letting $\epsilon:=\frac{3p}{2(M+p)}>0$, if one can approximate the single-sink Unsplittable Maximum Flow Over Time problem to a factor within $3/2-\epsilon$, then one can distinguish whether $\mathcal{I}$ is YES- or NO-instance, which is NP-hard in the undirected graph.
Thus, we complete the proof.
\end{proof}

\subsection{Maximum Confluent Flow Over Time}

This section gives the proof of the hardness of Confluent Maximum  Flow Over Time problem, i.e., Theorem~\ref{thm-confluentovertime}.

\begin{proof}[Proof of Theorem~\ref{thm-confluentovertime}]
We construct a directed network the same as Figure~\ref{fig:half-grid}(a), and let $T=M^2+NP+N+2$. If $\mathcal{I}$ embedded  in the network $\mathcal{G}$ is a YES-instance, then the maximum value of supplies sent to $t$ within time horizon $T$ is  at least $H_N (T-NP-N-2)=H_NM^2$. However, if $\mathcal{I}$ is a NO-instance, Lemma~\ref{lmm-treecut} shows there exists a cut whose weight is at most $2$, and hence the maximum value of supplies sent to $t$ within time horizon $T$ is at most $2T\approx 2 M^2$. Thus, if the single-sink Confluent Maximum  Flow Over Time problem in the directed graphs can be approximated to a factor within  $H_N/2$, one can determine whether $\mathcal{I}$ is a YES- or NO-instance, which is NP-hard in directed graphs.
By  setting those parameters the same as Theorem~\ref{thm-conlogappro}, we obtain the lower bound as desired.
\end{proof}




\section{Proofs of Constant Bicriteria Approximation Hardness of Dynamic Flows}
\label{App: Proofs of Constant Bicriteria Approximation Hardness}

\begin{figure}
\begin{center}
\vspace{-.0cm}\hspace{-0cm}{\includegraphics[scale=0.45]{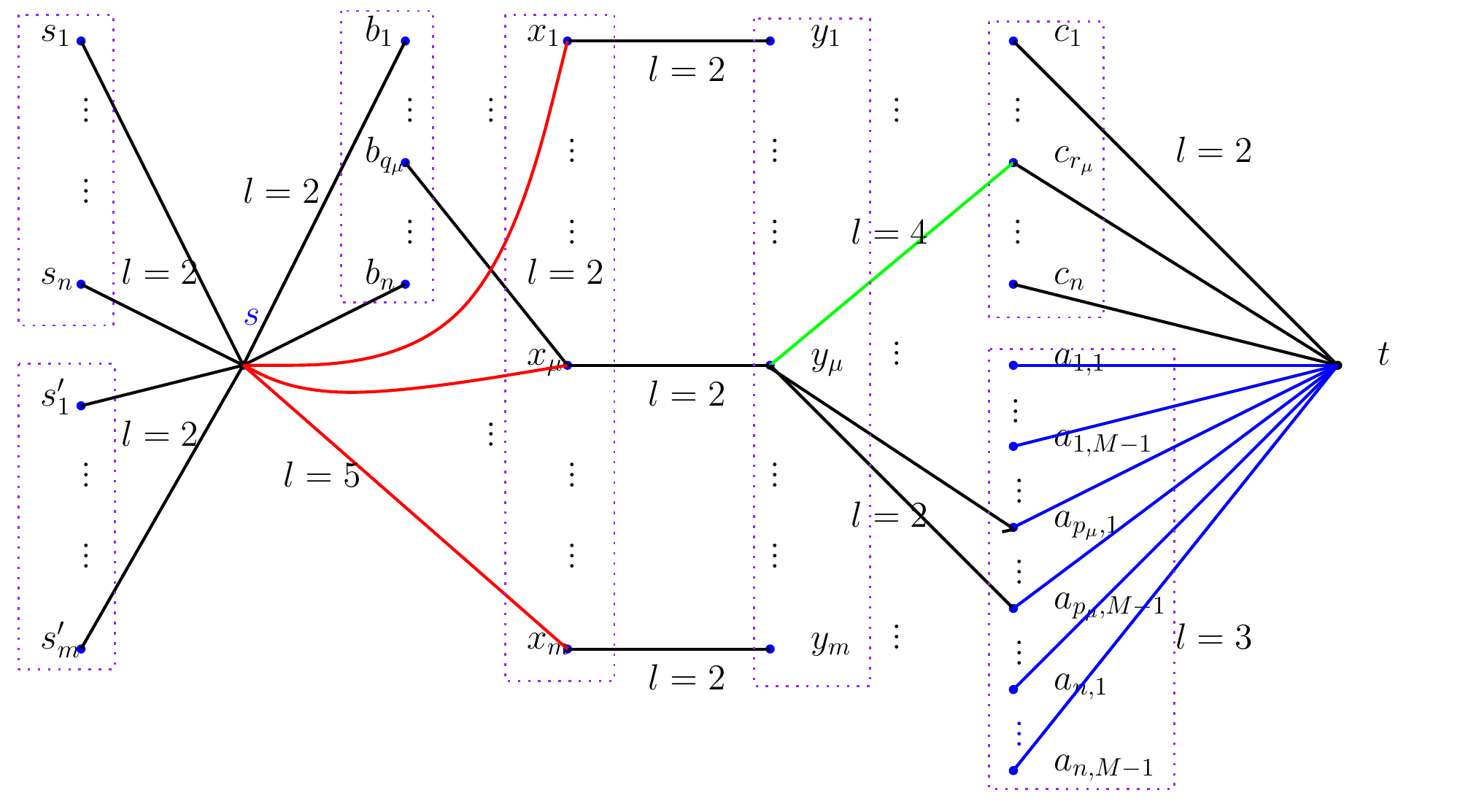}}\hspace{-0.8cm}
\vspace{-0.2cm}\caption{Hard instance for constant bicriteria approximation.}
\label{fig:bicriteria}\vspace{-0.8cm}
\end{center}
\end{figure}

First, we have the following lemma about the constructed instance.
\begin{lemma}\label{lmm-biapproxpath}
If the instance $\mathcal{I}$ is satisfiable, then $G$ contains $Mn$ edge-disjoint paths from sources $s_i$, $s'_{\mu}$ to $t$, whose length  are at most 14; if  $\mathcal{I}$ is $\epsilon_0$-unsatisfied,  then there are at most $(M-\epsilon_0/2)n$ such source-sink paths in $G$.
\end{lemma}
\begin{proof}
The proof is similar to that in~\cite{G2003}. The difference is that we add  $(n+m)$ sources and change all lengths. This does not change the existence or non-existence of edge-disjoint paths.
\end{proof}

Now, we prove Theorem~\ref{thm-biapproxhardQuick}.
\begin{proof}[Proof of Theorem~\ref{thm-biapproxhardQuick}]
We first consider the {\em Unsplittable} Maximum Flow Over Time problem.
Clearly, by Lemma~\ref{lmm-biapproxpath}, if the instance $\mathcal{I}$ is satisfiable, then $G$ contains $Mn$ edge-disjoint paths from sources $s_i$, $s'_{\mu}$ to $t$,  whose lengths  are at most 14. It implies that we can unsplittably send at least $Mn$ units of supplies to $t$ within the time horizon $T=14$.

However, if  $\mathcal{I}$ is $\epsilon_0$-unsatisfied,  then there are at most $(M-\epsilon_0/2)n$ source-sink paths whose lengths are bounded by 14 in $G$.
There are three kinds of paths with length at most 14:
\begin{enumerate}

\item {\bf $P_1$-path:} For any $z\in\{s_i:i\in[n]\}\cup\{s'_{\mu}:\mu\in[m]\}$ and any $\mu\in[m]$, $P_1=[z,s,b_{q_{\mu}},x_{\mu},y_{\mu},c_{r_{\mu}},t]$ has length 14;

\item {\bf $P_2$-path:} For any $z\in\{s_i:i\in[n]\}\cup\{s'_{\mu}:\mu\in[m]\}$, any $\mu\in[m]$ and any $l\in[M-1]$, $P_2=[z,s,b_{q_{\mu}},x_{\mu},y_{\mu},a_{p_{\mu}l},t] $ has length 13;

\item {\bf $Q$-path:} For any $z\in\{s_i:i\in[n]\}\cup\{s'_{\mu}:\mu\in[m]\}$, any $\mu\in[m]$ and any $l\in[M-1]$, $Q=[z,s,x_{\mu},y_{\mu},a_{p_{\mu}l},t] $ has length 14.

\end{enumerate}
Note also that any path with length larger than 14 must have length at least 15.

Suppose we are given the maximum set $S$ of edge-disjoint paths such that $|S|\leq (M-\epsilon_0/2)n$. Then, adding into $S$ any path $P'$ with length at most 14 would make $P'$ share at least one edge with some paths in $S$.
Note that $P'$  must be $P_1$-path, $P_2$-path or $Q$-path. Since the time horizon $T=14$ and the minimum length of any edge is 2, if one want to send the supply  to $t$ along $P'$ within the time horizon $T$, then the supply in $P'$  must be sent at time 0 (or time 1, for $P_2$-path), and cannot encounter other supplies in any edge at the same time (otherwise, the supply in $P'$ would be delayed by at least time 2).

We show in the following that  $P'$ must encounter other supplies carried by some paths in $S$.
Suppose $P'$ shares some edges with a path $J\in S$. Let $e=(u,v)$ be the first edge they share in $G$. Note that $J$ is $P_1$-path, $P_2$-path or $Q$-path. There are several cases:
\begin{enumerate}

\item [] {\bf Case~i:} If both  $J$ and $P'$ are the same kind of path, since the supplies carried by them must be sent at time 0, the supplies must reach  $e$ at the same time. Hence, at least one unit of supply cannot be sent to $t$ on time.

\item [] {\bf Case~ii:} If  $J$ is $P_1$-path and $P'$ is $P_2$-path (or, $P'$ is $P_1$-path and $J$ is $P_2$-path), then $e$ must be some edge no after $(x_{\mu},y_{\mu})$. Since $P_1$- and $P_2$-path have same sub-path pattern $[z,s,b_{q_{\mu}},x_{\mu},y_{\mu}]$, the supplies on $P'$ and $J$ must reach $e$ at the same time, resulting in at least one unit of supply cannot be sent to $t$ on time.

\item [] {\bf Case~iii:} If  $J$ is $P_i$-path and $P'$ is $Q$-path (or, $P'$ is $P_i$-path and $J$ is $Q$-path), where $i\in\{1,2\}$, then $e$ must be some edge no before $(x_{\mu},y_{\mu})$. Suppose $J$ is from source $z_1$ while $P'$ is from source $z_2$, where $z_1, z_2\in\{s_i:i\in[n]\}\cup\{s'_{\mu}:\mu\in[m]\}$. Note that the difference of length between the path from $z_1$ to $u$ and the path from $z_2$ to $u$ is 1. Also, the supply at $Q$-path must reach $u$ no earlier than $P_i$-path, and must wait for at least time 1 (since the edge $(u,v)$ will be occupied for time 2). Hence, the supply at $Q$-path cannot arrive at $t$ on time.

\end{enumerate}
Thus, one cannot send one more unit of supply  to $t$ along the path $P'$ within the given time horizon. This immediately implies that, if  $\mathcal{I}$ is $\epsilon_0$-unsatisfied, then one cannot send more than $(M-\epsilon_0/2)n$ units of supplies to $t$, within the time horizon $(\frac{15}{14}-\epsilon)T$, for any $\epsilon>0$. This result implies that it is NP-hard to obtain the $(1+\frac{\epsilon_0}{2M-\epsilon_0},\frac{15}{14}-\epsilon)$-approximation for the {\em Unsplittable} Maximum Flow Over Time problem in both directed and undirected graphs. (The analysis for directed graphs is similar.)

Now consider the {\em Confluent} Maximum Flow Over Time problem. To simplify the analysis, we remove all sources and $s,t$. We claim that after the node removal, those edge-disjoint paths with  length at most 14 in $G$ become node-disjoint. We now show this by contradiction. Suppose there are two edge-disjoint paths $Q_1$ and $Q_2$ with length at most 14 in $G$, and they share one node $v$.
\begin{enumerate}

\item If $v=b_{i}$ for some $i\in[n]$, then $Q_1$ and $Q_2$ share the edge $(s,b_i)$ in $G$ (contradiction).

\item If $v=x_{\mu}$ for some $\mu\in[m]$, then $Q_1$ and $Q_2$ share $(s,x_{\mu})$ (contradiction).

\item If $v=c_{i}$ for some $i\in[n]$, then $Q_1$ and $Q_2$ share $(c_i,t)$ (contradiction).

\item If $v=a_{p_{\mu} l}$ for some $\mu\in[m]$ and some $l\in[M-1]$, then $Q_1$ and $Q_2$ share $(a_{p_{\mu} l},t)$ (contradiction).

\item If $v=y_{\mu}$ for some $\mu\in[m]$, then $Q_1$ and $Q_2$ must share $(x_{\mu},y_{\mu})$, because otherwise one path of them must pass through some $a_{p_{\mu} l}$ in order to reach $y_{\mu}$, resulting in that the length exceeds 14 (contradiction).

\end{enumerate}

To show the lower bound of the confluent version, we need to slightly modify the graph $G$. We split the node $s$ into $m+n$ copies, and connect each copy to $s_i$ and $b_i$ (or, connect it to $s'_{\mu}$ and $x_{\mu}$) with edges of capacity 1 and length 2. Then, similar to the {\em edge-disjoint} paths, we can bound the number of the {\em node-disjoint} paths in both $\epsilon_0$-unsatisfied and satisfiable instance. Thus, applying the similar analysis for the unsplittable flow, we finally gives the desired approximation hardness for the Confluent Maximum Flow Over Time problem.
\end{proof}

The proof of Theorem~\ref{thm-biapproxhardTime} is similar, which we omit here.




\section{Proofs of Polylogarithmic Approximation for Single-Sink Confluent Dynamic Flows}
\label{App: Proofs of Polylogarithmic Approximation}

\subsection{Static Confluent Flows in Uncapacitated Networks with $\kappa$ Sources}\label{apd:uncapacitated_network}

In this section, we present an algorithm for finding a confluent flow in the uncapacitated network with $\kappa$ sources. In this problem, we are given a static directed $G = (V,A)$, where there are $\kappa$ sources $\{s_1,...,s_{\kappa}\}$ with the non-zero supply $d_i$ located at each $s_i (i\in [\kappa])$. Since the case of $\kappa=1 $  is trivial, we assume $\kappa\geq 2$. Here we consider the case that all non-zero supplies are uniform. Without loss of generality, we assume all supplies are unit, i.e., $d_1=...=d_\kappa=1$.
Also,
there exists a collection of sinks $\{t_1,...,t_k\}\subset V$. In addition, $G$ is an uncapacitated network, i.e., all edge  and node capacities are 1.

Before stating our algorithm, we give certain definitions regarding our $\kappa$-source setting.
\begin{definition}[Effective Length]
Given a network $G$  and a path $P$ in $G$, we define the effective
length of $P$, or simply $l_e(P)$, to be the number of sources in the path $P$ excluding the starting node.
\end{definition}
\begin{definition}[Effective Height]
Given a network $G$ and a subtree $T$ inside $G$, we define the effective
height of the tree, or simply $h_e(T)$, to be the maximum effective length over all leaf-to-root paths $P$ of $T$.
\end{definition}

Given a splittable flow $f$ that satisfies all supplies in $G$, we give a simple randomized rounding algorithm to get a confluent flow $f'$ from $f$:
\begin{enumerate}
\item []
For each
node $v\in V\setminus \{t_1,...,t_k\}$, select exactly one of its outgoing edges with probability of $f(e)/f^{out}(v)$, and let $e$ carry all flows out of $v$.
\end{enumerate}

Clearly, the resulting flow $f'$ is confluent. The selected edges together with the nodes in $V$ form a forest, where each tree is an arborescence directed toward a distinct sink $t_i$. Hence, the remaining work is to bound the node congestion of $f'$. Suppose we denote by $T_i$ the tree rooted at $t_i$. Then, the node congestion of $f'$ equals to the maximum number of {\em sources} in $T_i$ over all $i\in [k]$.

To analyze the congestion, we define the following random process $\mathscr{P}$.
\begin{enumerate}
\item []
Suppose  there is a directed acyclic graph (DAG) $D=(V, A)$ with a probability function $p(u,v)$ for each edge $(u,v)\in A$, satisfying that, for each node $u\in V$,  $\sum_{(u,v)\in A} p(u,v)\leq 1$. Let $\mathscr{P}(D)$ be the random process that each node $u$ selects at most one of its outgoing edges $(u,v)$ with probability $p(u,v)$.
\end{enumerate}
Note that the edges selected by $\mathscr{P}(D)$ form a forest. We denote by $\hat{N}_D(v)$  the number of {\em sources} in the subtree
rooted at $v$ in the forest, under the random process $\mathscr{P}(D)$.  Let $C_D(v)$ be the expectation of $\hat{N}_D(v)$, i.e., $C_D(v) = E[\hat{N}_D(v)]$. We have
\begin{eqnarray}
C_D(v)=\left\{
\begin{aligned}
&0 &   \text{$v$ is not a source and has no incoming edge} \\
&1 &   \text{$v$ is a source and has  no incoming edge} \\
&  1+\sum_{(u,v)\in D}p(u,v)C_D(u) &  \text{$v$ is a source and has incoming edges} \\
&  \sum_{(u,v)\in D}p(u,v)C_D(u) &  \text{otherwise}
 \end{aligned}
\right.
\end{eqnarray}

Now, let us induce the DAG from the given splittable flow $f$, as well as the probability function. Let $D=(V,A)$ be the DAG with $A$ as the set of edges carrying the flow of $f$, and let $p(u,v)=f(u,v)/f^{out}(u)$, for each $(u,v)\in A$.
It is easy to see that our randomized algorithm is equivalent to the random process $\mathscr{P}(D)$.  Then, the congestion of the resulting flow $f'$ on each node $v$ equals to $\hat{N}_D(v)$, and the congestion of $f'$ equals to $NC(f')=\hat{N}_D^*(v):=\max_{v} \hat{N}_D(v)$. Also, the congestion of $f$ equals to $NC(f)=C_D^*(v):=\max_{v} C_D(v)$.

To bound the congestion of $f'$, we only need to bound the random variable $\max_{v} \hat{N}_D(v)$ for the given DAG $D$. Thus, we turn to analyze the random process. The idea is simple: We upper bound all the moments of the random variable $\hat{N}_D(v)$ for each $v\in V$,  and utilize the Markov's inequality to guarantee that, with high probability, $\hat{N}_D(v)$ slightly deviates from $C_D^*(v)$.

\subsubsection{Bounding Effective Height}

We start from bounding the effective height of the tree formed by the resulting flow $f'$.

Given a source node $v$, the {\em effective} distance from $v$ to the root of the subtree containing $v$ in $\mathscr{P}(D)$ is the effective length of the random walk starting from $v$ on $D$, according to the probability function $p$.
Suppose a random walk will reach a node $u$ in its next step (or hop). If $u$ is a source in $D$, we call the step (or hop) as {\em effective}. Clearly, in a path from a node $v$ to another node $u$, the number of effective steps is equal to the effective length of the path.

For any source $u$, let $P(u;i)$ be the probability that the random walk start from $v$ and reaches the non-sink node $u$ after $i$ {\em effective steps}. Then, we have this recurrence relation:
\begin{eqnarray}
P(u;i)=\left\{
\begin{aligned}
&0                                 & \hspace{1cm}  \text{$i=0$ and $u\neq v$} \\
&1                                 & \hspace{1cm}  \text{$i=0$ and $u= v$} \\
& \sum_{(w,u)\in D}P(w;i-1)p(w,u)  & \hspace{1cm} \text{$i>0$ and $u$ is a source}\\
& \sum_{(w,u)\in D}P(w;i)p(w,u)    & \hspace{1cm} \text{$i>0$ and $u$ is not a source}
 \end{aligned}
\right.
\end{eqnarray}

\begin{lemma}\label{lmm-NonSinkProb}
For any non-sink node $u$ and $i\geq 0$, $P(u;i)\leq \max(C_D(u),1)(1-1/C_D^*)^i$.
\end{lemma}
\begin{proof}
We prove by induction on $i$. For $i = 0$ the claim is trivially true.
To prove the induction step, consider two cases:
\begin{enumerate}

\item [$\bullet$]
If $u$ is a source, we have
\begin{eqnarray}
P(u;i) &=& \sum_{(w,u)\in D}P(w;i-1)p(w,u)\nonumber\\
&\leq&  (1-1/C_D^*)^{i-1} \sum_{(w,u)\in D}C_D(w)p(w,u)\nonumber\\
&\leq&  (1-1/C_D^*)^{i-1} (C_D(u)-1)\nonumber\\
&\leq&  (1-1/C_D^*)^{i} C_D(u)\hspace{3cm} \text{(since $C_D(u)\leq C_D^*$)}\nonumber\\
&\leq& \max(C_D(u),1)(1-1/C_D^*)^i\nonumber
\end{eqnarray}

\item [$\bullet$]
If $u$ is not a source, we have
\begin{eqnarray}
P(u;i) &=& \sum_{(w,u)\in D}P(w;i)p(w,u)\nonumber\\
&\leq&  (1-1/C_D^*)^{i} \sum_{(w,u)\in D}C_D(w)p(w,u)\nonumber\\
&\leq&  (1-1/C_D^*)^{i} (C_D(u)-1)\nonumber\\
&\leq&  (1-1/C_D^*)^{i} C_D(u)\nonumber \\
&\leq& \max(C_D(u),1)(1-1/C_D^*)^i\nonumber
\end{eqnarray}

\end{enumerate}
Thus, we complete the proof.
\end{proof}

\begin{lemma}\label{lmm-height}
The effective height of any tree in $\mathscr{P}(D)$ is at most $O(C_D^* \log \kappa)$ with
probability of at least $1-\kappa^{-c}$ for a large positive constant $c$.
\end{lemma}
\begin{proof}
Suppose the random walk from $v$ reaches a non-sink node $u$ after $\alpha C_D^*\ln (\kappa C_D^*)$ effective steps. If $u$ is not a source,  we can back-track to the last source, denote as $u'$, along the path in the random walk. Note that the random walk from $v$ reaches $u'$ after $\alpha C_D^*\ln (\kappa C_D^*)$ effective steps.

Since $\kappa>1$, we have $C_D^*>1$. By Lemma~\ref{lmm-NonSinkProb}, the probability that the random walk from $v$ reaches a source $u$ (or $u'$ if $u$ is not a source) after $\alpha C_D^*\ln (\kappa C_D^*)$ effective steps is at most $\kappa C_D^*(1-1/C_D^*)^{\alpha C_D^* \ln (\kappa C_D^*)}\leq \kappa^{1-\alpha}$. Thus, the random walk terminates at a sink in $\alpha C_D^*\ln (\kappa C_D^*)$ effective steps with the probability at least $1-\kappa^{1-\alpha}$. This means that the effective height of any tree in $\mathscr{P}(D)$ is at most $O(C_D^* \log \kappa)$ with
probability of at least $1-\kappa^{-c}$ for a large positive constant $c$.
\end{proof}

Suppose each node is a source. Then, the effective height of a tree equals to its height. Thus, by setting $\kappa=n$, we bound the {\em height} of any tree in $D$, and conclude as Lemma~\ref{lmm-mergetree}.

\subsubsection{Bounding the Moment of $\hat{N}_D(v)$}

Instead of directly bounding the moment of $\hat{N}_D(v)$ in $D$, we would like to transform $D$ into the random process on some simpler tree graphs.

We transform $D$ into a collection of trees, denote as $\mathcal{T}$, through a sequence of steps. Let $D_j$ denote the DAG obtained after $j$ step $(j\geq 0)$, where $D_0=D$. For a given DAG $D$, let $D(v)$ denote the subgraph of $D$ induced by all of the nodes that
can reach $v$ in $D$. Step $j + 1$ proceeds as follows:
\begin{enumerate}

\item [1.]
Find a node $v\in D_j$ such that the subgraph $D_j(v)$ is a tree and $v$ has more than one outgoing  edges. If no such node is found, then $D_j$ is a tree and the transformation is completed.

\item [2.]
Let $(v,u_1),..., (v,u_k)$ denote the $k \geq 2$ edges going out of $v$. We transform $D_j$ into $D_{j+1}$ as follows.
Replace $v$ and the subtree $D_j(v)$ rooted at $v$ by $k$ copies of each, and replace the edge $(v,u_i)$ by
the edge $(v_i,u_i)$, where $v_i$ is the $i$-th copy of $v$. Each new edge inherits the probability of the edge it replaces or copies.

\end{enumerate}

Now we introduce some notations. For a given node $u$ and a nonnegative integer $h_e$, let the random variable $\hat{N}_{D_j}(u, h_e)$ denote the number of sources within $h_e$ effective hops of $u$, in the
subtree rooted at $u$ under the random process $\mathscr{P}(D_j)$. We note that $E[\hat{N}_{D_j}(u, h_e)]=E[\hat{N}_{D_{j+1}}(u, h_e)]$.
This implies the following equality $C_D^*=C_{\mathcal{T}}^*$.

{\em Contraction.} Now, we introduce a new process---contraction:
\begin{enumerate}

\item [1.]
For a given node $u$ in the DAG $D_j\ (j\geq 0)$, let $S:= \{u\}\cup \{s_1,...,s_\kappa\}\cup \{t_1,...,t_k\}$, i.e., the node set that includes all sources, all sinks and the node $u$. Here, we view a non-source node $u$ as a sink in our analysis.

\item [2.]
For any node $v\in V\setminus S$, let $(u_1,v),...,(u_p,v)$ be all incoming edges of $v$, and $(v,w_1),...,(v,w_q)$ be all outgoing edges of $v$.

\item [3.]
Remove $v$. For each $u_i$ and any $i\in [p]$, add new edges $(u_i,w_1),...,(u_i,w_q)$, and set the probability $p(u_i,w_j):=p(u_i,v)p(v,w_j)$ for $j\in [q]$. This is illustrated in Figure~\ref{fig:contract}.

\item [4.]
Repeat the above process until all nodes in $V\setminus S$ are contracted.

\end{enumerate}
Denote the DAG $D_j$ after  contraction by $D_j^c$. Clearly, the random process $\mathscr{P}(D_j)$ is equivalent to $\mathscr{P}(D_j^c)$, and
we have $E[(\hat{N}_{D_j^c}(u,h_e))^r]=E[(\hat{N}_{D_j}(u,h_e))^r]$ for any integer $r$. Also, it implies that $C_{\mathcal{T}}^*\geq C_{\mathcal{T}^{c}}^*$, since we contract some node in $\mathcal{T}$.

\begin{figure}
\begin{center}
\vspace{-0cm}\hspace{-0cm}{\includegraphics[scale=0.6]{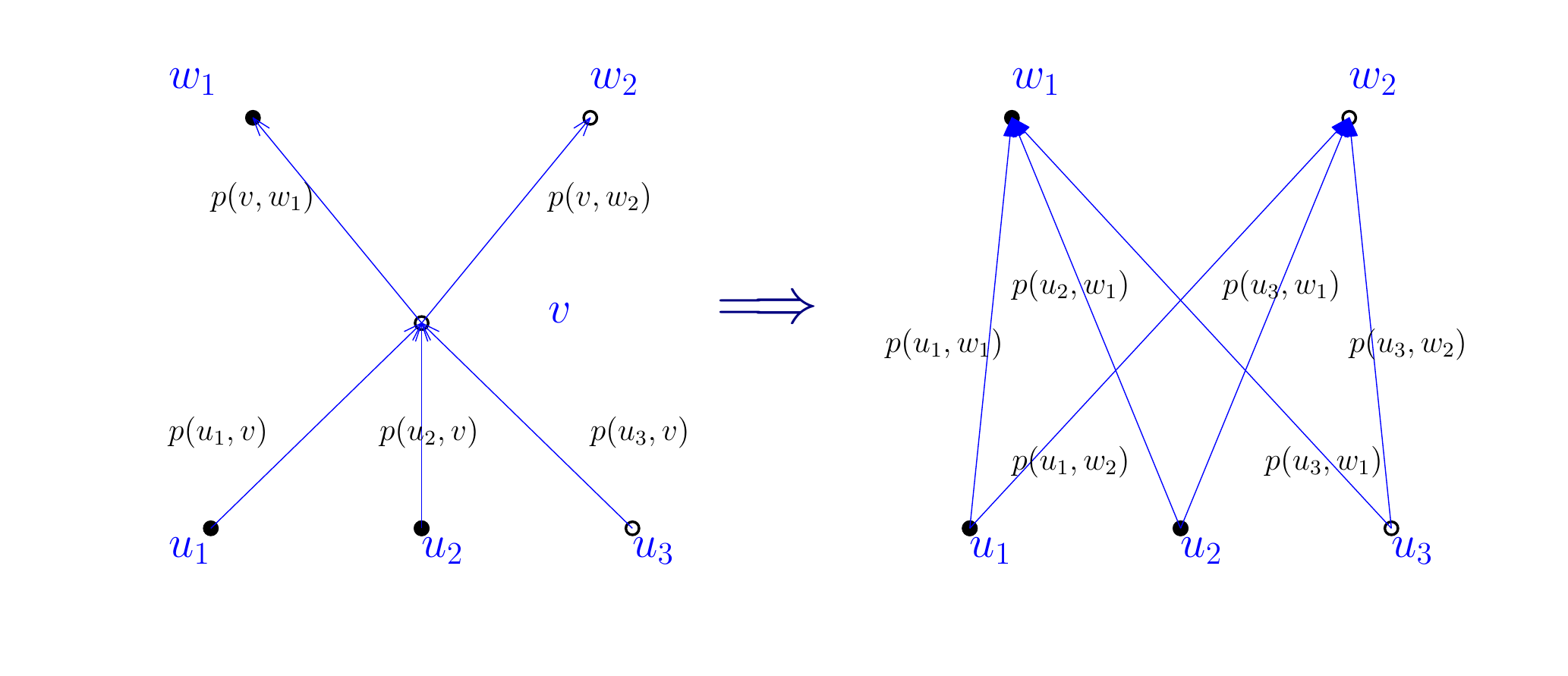}}
\\ \caption{An example of contraction on the node $v$, where $p(u_i,w_j)=p(u_i,v)p(v,w_j)$ with $i=1,2,3$ and $j=1,2$. Black nodes represent nodes in $S$, while white nodes represent nodes in $V\setminus S$.}\vspace{-0.8cm}
\label{fig:contract}
\end{center}
\end{figure}

Note that after contraction, all nodes except sinks have unit supplies, and each hops in random walk is an effective hops. For a node $u$ in a given DAG $D$ and any nonnegative integer $h$, let   $N_{D}(u,h)$ denote the number of {\em nodes}  within $h$ hops of $u$, in the
subtree rooted at $u$ under the random process $\mathscr{P}(D)$. Then, we have $\hat{N}_{D_j^c}(u,h_e)=N_{D_j^c}(u,h_e)$.

Thus, we can
bound the moment of $\hat{N}_{D_j}(u,h_e)$.
\begin{lemma}\label{lmm-HigherMoment}
For integers $h, r\geq 0$, and node $u$, we have $E[(\hat{N}_{D_j}(u,h_e))^r]\leq E[(\hat{N}_{D_{j+1}}(u,h_e))^r]$.
\end{lemma}
\begin{proof}
%
Due to Lemma~4.2 in \cite{merge03}, which assumes all nodes have unit supplies, we have
$$
E[(N_{D_j^c}(u,h_e))^r]\leq E[(N_{D_{j+1}^c}(u,h_e))^r],
$$
since in the DAG $D_j^c$ all nodes except sinks are sources.

It immediately means that
\begin{eqnarray}
&&E[(\hat{N}_{D_j}(u,h_e))^r]= E[(\hat{N}_{D_j^c}(u,h_e))^r]
=E[(N_{D_j^c}(u,h_e))^r]\nonumber\\
&\leq& E[(N_{D_{j+1}^c}(u,h_e))^r]
=E[(\hat{N}_{D_{j+1}^c}(u,h_e))^r]\nonumber\\
&=&E[(\hat{N}_{D_{j+1}}(u,h_e))^r].\nonumber
\end{eqnarray}
\end{proof}

\begin{lemma}\label{lmm-FinalCong}
For any $i\geq 0$, $E[(\hat{N}_{\mathcal{T}^c}(u,h_e))^i]$ is at most $\frac{i!(C_{\mathcal{T}^c}^* )^i(h_e O(\log \kappa))^{i-1}}{2^{i-1}}$.
\end{lemma}
\begin{proof}
Note that $\hat{N}_{\mathcal{T}^c}(u,h_e)=N_{\mathcal{T}^c}(u,h_e)$. Thus,
due to Corollary~4.5.1 in~\cite{merge03}, the lemma follows. Also note that the transformation inside Corollary~4.5.1 would introduce a $O(\log \kappa)$ factor in the tree height.
\end{proof}

\subsubsection{Proof of Theorem~\ref{thm-mergetree}}

This section completes the proof of Theorem~\ref{thm-mergetree}. To bound the congestion of $f'$, we only need to bound the number of sources in the
subtree rooted at each sink $t_j$ under the random process $\mathscr{P}(D)$. Due to Lemma~\ref{lmm-height},  the effective height of any tree in $\mathscr{P}(D)$ is $O(C_D^*\log \kappa)$ with probability at least $1-\kappa^{-c}$. It suffices to bound $\hat{N}_{D}(t_j,h_e)$, where $h_e=O(C_D^*\log \kappa)$.

We use Markov's inequality
\begin{eqnarray}
&& Pr[\hat{N}_{D}(t_j,h_e) > \alpha h_e C_D^*\cdot O(\log \kappa)] \nonumber\\
&=& Pr[(\hat{N}_{D}(t_j,h_e) )^i> \alpha^i h_e^i (C_D^* \cdot O(\log \kappa))^i]\nonumber\\
& \leq & \frac{E[(\hat{N}_{D}(t_j,h_e) )^i]}{\alpha^i h_e^i (C_D^* \cdot O(\log \kappa))^i}\nonumber\\
& \leq & \frac{E[(\hat{N}_{\mathcal{T}}(t_j,h_e) )^i]}{\alpha^i h_e^i (C_D^* \cdot O(\log \kappa))^i} \hspace{2cm} \text{(due to Lemma~\ref{lmm-HigherMoment})}\nonumber\\
& \leq & \frac{i!(C_{D}^* )^i(h_e\cdot O(\log \kappa))^{i-1}}{2^{i-1}\alpha^i h_e^i (C_D^* \cdot O(\log \kappa))^i} \hspace{1.35cm} \text{(due to Lemma~\ref{lmm-FinalCong} and $C_{D}^*\geq C_{\mathcal{T}^c}^*$)}\nonumber\\
& \leq & \frac{i!}{2^{i-1}\alpha^i h_e \cdot O(\log \kappa)}. \nonumber
\end{eqnarray}
Now, by setting $\alpha=i=O(\log \kappa)$, the probability above is at most $\kappa^{-r}$ for a large positive constant $r$. It immediately means that the congestion of $f'$ is at most $O(NC(f)^2\log^3 \kappa)$ (since $NC(f)=C_{D}^*$), with probability of at least $1-\kappa^{-r}$.

Since $\kappa\geq 2$, we have $1-\kappa^{-r}\geq 1/2$. We repeat our randomized rounding algorithm for $O(\log n)$ iterations, which yields a confluent flow with desired congestion with high probability. Therefore, we complete the proof of Theorem~\ref{thm-mergetree}.

\subsection{Static Length-Bounded Confluent Flows in Monotonic Networks}\label{apd:monotonic_network}

This section presents the proof of Theorem~\ref{thm-monotonicflow} and \ref{thm-monotonicflow-relaxed}.

{\em Construction of sub-networks.} Without loss of generality, we assume the minimum node capacity is 1. We first process the original network: (i) Round up each node capacity to the nearest power of $\log^4 n$; (ii) Partition the nodes into groups according to their capacities, such that the nodes in the same group have  the same capacity; (iii) For any edge $(u,v)$ with $c(u)=\log^{4i} n$ and $c(v)=\log^{4(i+r)} n$ $(r>1)$, add dummy nodes $u^{(j)}$ of $u$ with capacity of $\log^{4(i+j)} n$ where $1\leq j<r$, and replace $(u,v)$ with a path connecting $u,u^{(1)},u^{(2)},...,u^{(r-1)},v$ sequentially (all dummy nodes have zero supplies).
%
%
%
%
%
%

Given a 1-satisfiable flow $f$ with node congestion at most 1, we now construct the $i$-th sub-network $G_i:=(V_i,A_i)$ as follows:
\begin{enumerate}

\item [$\bullet$] Set $V_i$ as the set of nodes of capacity $\log^{4i} n$ and  their  incident nodes  of capacity $\log^{4(i+1)} n$;

\item [$\bullet$] Let $S_i$ be the set of nodes of capacity $\log^{4(i+1)} n$, signifying the sink set in $G_i$;

\item [$\bullet$] Let $A_i$ consist of all arcs  induced by $V_i$ in $G$, excluding all arcs between nodes in $S_i$;

\item [$\bullet$]  Set the supply $d_i(v)$ at the node $v\in V_i\setminus S_{i}$ as $d_i(v):=d(v)+\sum_{u:u\in V_{i-1}\setminus S_{i-1}}f(u,v)$.

\end{enumerate}

Conceptually, we would partition nodes into $r:=\lceil \log c_{\max} /(4\log\log n)\rceil$ groups, where $c_{\max}$ is the maximum node capacity, and then $r$ sub-networks, the size of which might not be polynomial. However, because of the fact that the total number of nodes  is at most $n$ (excluding those dummy nodes whose total size will be reduced to polynomial later), there exist at most $n$ sub-networks. Hence, the decomposition of network can be done in  polynomial time.

Note that $f$ induces a flow $f_i$  with node congestion at most 1 in $G_i$. For each sub-network $G_i$: as $f_i$ might not be confluent, we would like to round it to a confluent flow $f'_i$. Because all non-sink nodes have uniform capacity due to our construction, we then can view the sub-network as an uncapacitated network. On the other hand, since the node congestion of $f_i$ is at most 1 in $G_i$, the supply $d_i(v)$ at each non-sink node $v$ is at most the node capacity  $\log^{4i} n$, and now we can view all supplies as $\log^{4i} n$ (i.e., uniform-supply case). Thus,
we can utilize Theorem~\ref{thm-mergetree} to compute a confluent flow $f'_i$ from $f_i$.

Remember that those rounding processes in each $G_i$ are independent from each other. In the following, we  bound the node congestion of $f'_i$.
\begin{lemma}\label{lmm-sepcongestion}
Whp, the node congestion at $v \in \mathcal{S}_i$ is at most $1 + \frac{O(1)}{\log n}$ in the confluent flow $f'_i$.
\end{lemma}
\begin{proof}
Consider the $i$-th sub-network, where all non-sink nodes have capacity $\log^{4i} n$ and sink nodes have capacity of $\log^{4(i+1)} n$. Also,
because the node congestion of $f_i$ is at most 1 in $G_i$, according to our construction, the supply $d_i(v)$ at each non-sink node $v$ is at most the node capacity  $\log^{4i} n$. Hence, we scale all node capacities and supplies down by a $\log^{4i} n$ factor. Note that the scaling does not affect the congestion. We can view the new network as  an ``uncapacitated'' network with unit supplies
(ignoring the capacity and supplies of {\em sinks in $S_i$} at this stage).

Now, we can utilize Theorem~\ref{thm-mergetree} to round the flow $f_i$ to a confluent flow $f'_i$ in $G_i$, which results in the node congestion $O(NC(f_i)^2\log^3 n)=O(\log^3 n)$ in the ``uncapacitated'' network $G_i$ whp. This implies the rounding process adds the flow  of value at most $O(\log^3 n)$ at any sink $v\in S_i$ (this flow increment only results from the supplies of non-sink nodes, and we have not yet taken into account those supplies at sinks in $S_i$).

On the other hand, note that after the scaling via Theorem~\ref{thm-mergetree}, for each sink $v\in S_i$, the initial node congestion induced by the original flow $f_i$ is at most 1, which means the supplies at $v$ is at most equal to its capacity $\log^{4} n$. Together with the flow increased by rounding, the total flow located at $v$ is at most $\log^{4} n+O(\log^{3} n)$. Thus, it follows that, whp, the node congestion at $v$
\begin{eqnarray}
NC(v)\leq \frac{\log^{4} n+O(\log^{3} n)}{\log^{4} n}=1+\frac{O(1)}{\log n}.\nonumber
\end{eqnarray}
\end{proof}

Now, we have a collection of confluent flows $\{f'_0,f'_1,...,f'_r\}$ in the  sub-networks $\{G_0,G_1,G_2,...,G_r\}$, where $r=\lceil \log c_{\max} /(4\log\log n)\rceil$.
Obviously, linking the support of those flows $\{f'_0,f'_1,...,f'_r\}$ induces a confluent flow $\mathbbm{f}$ for routing all supplies in the original network $G$. However, the node congestion of $\mathbbm{f}$ might be large, because the flows $f_1,...,f_r$ are turned into $f_1',...,f_r'$ with larger congestion  than before.

To bound the congestion of $\mathbbm{f}$, let $\mathbbm{f}_i$ be the induced confluent flow by $\{f'_0,f'_1,...,f'_i\}$ on $V_0\cup V_1\cup \cdots \cup V_i$. Observe that the congestion of node $v\in V_i$ in $\mathbbm{f}_i$ is the same as its congestion in $\mathbbm{f}$. It suffices to analyze the node congestion of $v\in V_i$ in $\mathbbm{f}_i$.
\begin{lemma}\label{lmm-mergecongestion}
For any $i\in\{0,1,...,r\}$, the congestion of $v\in V_i$ is at most $O(\log^4 n)$ in the confluent flow $\mathbbm{f}_i$ whp.
\end{lemma}
\begin{proof}
Consider any sink $v^{(i)}\in S_i\subset V_i$.  Before linking the flows, $f_i$ has congestion at most 1 in $G_i$, and hence the flow $f_i$ induces the supply of $\log^{4i} n$ at each non-sink node in $V_i$. By Lemma~\ref{lmm-sepcongestion}, we know, whp, the rounding of $f_i$ would induces at most $(1+O(1)/\log n)\log^{4(i+1)} n$ units of new flow load at $v^{(i)}$.
Note that this is based on the fact that $f_i$  only induces the supply of at most $\log^{4i} n$ at each non-sink node (up to now, we only consider the sub-network separately).

However, after linking those flows, the congestion of $v$ is caused by the supplies induced by $\mathbbm{f}_{i-1}$, rather than $f_{i}$. Suppose the sink $v^{(i-1)}\in S_{i-1}$ has congestion at most $NC(v^{(i-1)})$ in $\mathbbm{f}_{i-1}$. This would bring in the supply of $NC(v^{(i-1)})\log^{4i} n$ at those non-sink nodes in $V_i$. 
Hence, noting that the supply of {\em non-sink} node is $\log^{4i} n$  before linking (as mentioned above), the supply of {\em non-sink} node in $V_i$ would be larger than the supply before linking by an $NC(v^{(i-1)})$ factor. It implies that the flow load at any sink $v^{(i)}\in S_i$ after rounding would be also increased by at most $NC(v^{(i-1)})$. Thus, whp,
the congestion of $v^{(i)}\in S_i$ is
\begin{eqnarray}\label{eqn-congestion}
NC(v^{(i)})\leq \frac{NC(v^{(i-1)}) (1+O(1)/\log n)\log^{4(i+1)} n}{\log^{4(i+1)} n}\leq NC(v^{(i-1)}) (1+O(1)/\log n).
\end{eqnarray}
Recall that the congestion of any sink in $V_0$ is at most $(1+O(1)/\log n)$, namely $NC(v^{(0)})\leq (1+O(1)/\log n)$. Recursively utilizing Inequality~\ref{eqn-congestion} yields $NC(v^{(i)})\leq (1+O(1)/\log n)^{i+1}$.

Now there exists two cases: (a) $0\leq i\leq 2\log n$, and (b)  $2\log n<i\leq r$. In Case~(a), it is clear that the congestion $NC(v^{(i)})\leq O(1)$. In Case~(b), observe that the capacity, and hence the supply, of
nodes in groups $\{f'_0,f'_1,...,f'_{i-2\log n-1}\}$ are at most a factor $1/n^2$ of the capacity in $S_i$.
Since  at most $n$ nodes have non-zero supplies, these nodes' contribution to $NC(v^{(i)})$ is negligible. Thus, one can truncate those flows into $\{f'_{i-2\log n},f'_{i-2\log n+2},...,f'_{i}\}$, and thus the congestion brought by them is bounded by $NC(v^{(i)})\leq O(1)$ as above. Hence, we bound the congestion of sinks in each $V_i$.

On the other hand, we also need to bound the congestion of non-sink nodes in each $V_i$. Since any non-sink node has capacity smaller than the sink in $V_i$ by a $\log^4 n$ factor, while its load is at most as large as the load of sink, where its load is sent, then its congestion is at most $O(\log^4 n)$ times larger than the sink. Thus, we complete the proof.
\end{proof}

This  immediately yields the following lemma:
\begin{lemma}\label{lmm-finalcongestion}
Whp, the node congestion of the confluent flow $\mathbbm{f}$ is at most $O(\log^4 n)$.
\end{lemma}
Up to now, we have guaranteed the computed confluent flow $\mathbbm{f}$ has small node congestion. More importantly, our method has one more advantage: The length of any path in the resulting confluent flow $\mathbbm{f}$ is bounded.
\begin{lemma}\label{lmm-flowlength}
For any path in the confluent flow $\mathbbm{f}$, the length is at most $O\left( \frac{L(f)\log n \log c_{\max}}{\log\log n} \right)$ whp.
\end{lemma}
\begin{proof}
First, note that we utilize Theorem~\ref{thm-mergetree} to round $f_i$ to the confluent flow $f'_i$, whose support is a tree. As $f_i$ is induced by $f$ in $G_i$, it is clear that the flow length of $f_i$ is at most that of  $f$, i.e., $L(f_i)\leq L(f)$. Also,
noting that the randomized rounding process in Theorem~\ref{thm-mergetree} is based on the support of $f_i$, we know the length of any edge in the resulting tree is not larger than $L(f_i)$. Furthermore, by Lemma~\ref{lmm-mergetree}, the height of the tree support of $f'_i$ is at most $O(\log n)$ whp. It immediately implies that any path length in $f'_i$ is at most $O(L(f)\log n)$ whp.

On the other hand, to construct the final confluent flow $\mathbbm{f}$, we link the support of those confluent flows $f'_0,...,f'_r$, where $r=\lceil \log c_{\max} /(4\log\log n)\rceil$. As there exists at most $r$ layers, it means any path length in $\mathbbm{f}$ would be bounded by $O(rL(f)\log n)$ as desired.
\end{proof}

The remaining thing needed to be dealt with is those introduced dummy nodes in (iii). When the maximum capacity  $c_{\max}$ is large, there would be $O(\log c_{\max})$ dummy nodes. However, one can just introduce a single dummy node $u'$ of capacity $\log^{4(i+1)} n$, and, after obtaining $\mathbbm{f}_i$, contract $u'$ into $v$ as the other dummy nodes are superfluous. This renders the sub-networks of polynomial size as desired. Thus, combining Lemma~\ref{lmm-finalcongestion} and~\ref{lmm-flowlength}, and noting that the rounding up of each node capacity in (i) also induces a $O(\log^4 n)$ factor in congestion,  we conclude as stated in Theorem~\ref{thm-monotonicflow}.

Now we show how to restrict our technique to the confluent routing where the length-bounded constraint is removed, hence proving Theorem~\ref{thm-monotonicflow-relaxed}.

\begin{proof}[Proof of Theorem~\ref{thm-monotonicflow-relaxed}]
Again, we assume the minimum node capacity is 1.
First, we modify the construction of sub-networks such that all capacities depend on $\kappa$:
(i) Round up each node capacity to the nearest power of $\log^4 \kappa$; (ii) Partition the nodes into groups according to their capacities, such that the nodes in the same group have  the same capacity; (iii) For any edge $(u,v)$ with $c(u)=\log^{4i} \kappa$ and $c(v)=\log^{4(i+r)} \kappa$ $(r>1)$, add dummy nodes $u^{(j)}$ of $u$ with capacity of $\log^{4(i+j)} \kappa$ where $1\leq j<r$, and replace $(u,v)$ with a path connecting $u,u^{(1)},u^{(2)},...,u^{(r-1)},v$ sequentially (all dummy nodes have zero supplies).

Then, we would like to find an unsplittable flow in $G$.
We will utilize the following theorem:
\begin{theorem}[\cite{dinitz1999single}]\label{thm-unsplitflow}
Let $G=(V,E)$ be an edge-capacitated directed graph with a single sink $t$ and $k$ sources $s_i$ with supplies $d_i\   ( i\in [k])$. If there is a feasible flow for routing all supplies to $t$, and $G$ satisfies $\max_{i\in[k]} d_i\leq \min_{e\in E} c(e)$, then there exists a polynomial-time algorithm for computing an unsplittable flow satisfying all supplies with edge congestion at most 2.
\end{theorem}

In fact, Theorem~\ref{thm-unsplitflow} only works for edge-capacitated networks. However, we can use it to find an unsplittable flow in the   node-capacitated network. Indeed, we can induce an edge-capacitated network $G'$ from the original node-capacitated network $G$ in the following: (a) Divide each node $v$ into $v_{in}$ and $v_{out}$ with capacity $c(v)$; (b) Connect $v_{in}$ with all arcs going into $v$, and $v_{out}$ with all arcs going out of $v$, and add a new arc $(v_{in},v_{out})$;
(c) Let each arc $e=(u,v)$ have the capacity $c(e)=\min(c(u),c(v))$.

Obviously, if there is a 1-satisfiable  flow with node congestion at most 1 in $G$, there must be a 1-satisfiable flow with edge congestion at most 1 in $G'$. If $G$ satisfies the no-bottleneck assumption $\max_{i\in[k]} d_i\leq \min_{v\in V} c(v)$, $G'$ must satisfy $\max_{i\in[k]} d_i\leq \min_{e\in E} c(e)$. Now, we can apply Theorem~\ref{thm-unsplitflow} to $G'$ for finding an unsplittable flow $\zeta$ with edge congestion at most 2 in $G'$. Consider the arc $e=(v_{in},v_{out})$ for each $v\in V$. This arc has congestion at most 2, which means the node congestion of $v$ is at most 2. Thus, we can use $\zeta$ to induce back an unsplittable flow $f$ with node congestion at most 2 in $G$.

Given a 1-satisfiable {\em unsplittable} flow $f$ with node congestion at most 2, we now construct the $i$-th sub-network $G_i:=(V_i,A_i)$ as follows (similar to before):
\begin{enumerate}

\item [$\bullet$] Set $V_i$ as the set of nodes of capacity $\log^{4i} \kappa$ and  their  incident nodes  of capacity $\log^{4(i+1)} \kappa$;

\item [$\bullet$] Let $S_i$ be the set of nodes of capacity $\log^{4(i+1)} \kappa$, signifying the sink set in $G_i$;

\item [$\bullet$] Let $A_i$ consist of all arcs  induced by $V_i$ in $G$, excluding all arcs between nodes in $S_i$;

\item [$\bullet$]  Set the supply $d_i(v)$ at the node $v\in V_i\setminus S_{i}$ as $d_i(v):=d(v)+\sum_{u:u\in V_{i-1}\setminus S_{i-1}}f(u,v)$.

\end{enumerate}
Note that {\em the number of non-zero supplies is $O(\kappa)$ in each $G_i$}.  Because we utilize an unsplittable flow, and there are at most $\kappa$ flow paths from $V_{i-1}\setminus S_{i-1}$ to $V_{i}\setminus S_{i}$, which means the number of the extra sources induced by flows between sub-networks is at most $\kappa$.

Finally, we apply similar analysis as shown in the proof of Theorem~\ref{thm-monotonicflow}, and bound the node congestion by $O(\log^8 \kappa)$ whp. Thus, we complete the proof.
\end{proof}


\subsection{Static Length-Bounded Confluent Flows in General Networks}\label{Apd:StaticConFlow}

\begin{figure}
\begin{center}
\vspace{-.0cm}\hspace{-0cm}{\includegraphics[scale=0.6]{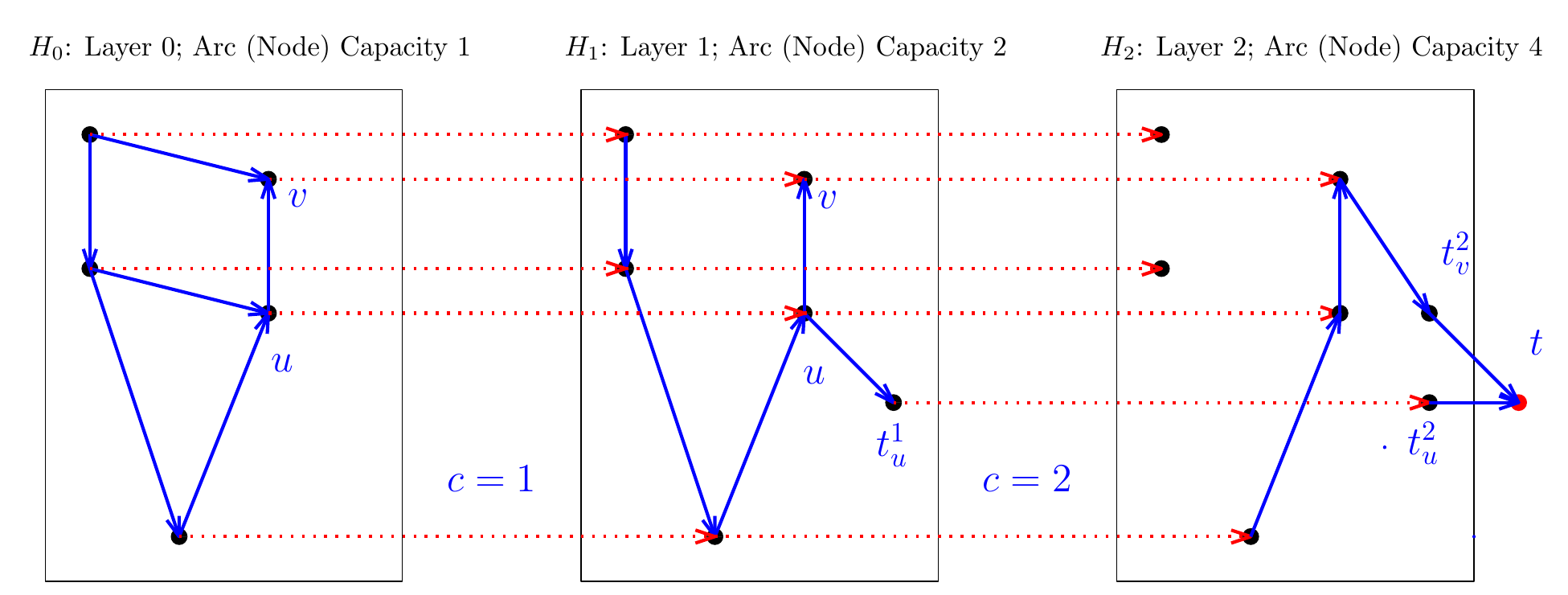}}
\caption{An example of a 3-layered network.}\vspace{-0.8cm}
\label{fig:3layer}
\end{center}
\end{figure}

We give the missing lemmas and proofs in Section~\ref{sec:StaticConFlow}.

Note that our multi-layer network construction differs from~\cite{shepherd2015polylogarithmic} in three main places. (i) Our multi-layer network takes into account bounding the length of the resulting flow.
(ii) Our technique works for the routing of confluent flow in {\em edge-capacitated} network. 
Moreover, because of (ii), our vertical arcs are contained in the $i$-th layers of $H$ when their original edge capacities in $G$ are at least $2^i$, namely depending on their own capacities, rather than depending on the capacities of their incident nodes. 
(iii) Another big difference is that we introduce dummy sinks so as to bound the edge congestion of confluent flow.
These facts allow us to have the following lemmas.

First, we show  there is a polynomial-time scheme for re-routing $h$, which is a confluent flow $h$ in the $k$-layer network $H$, into a confluent flow in the original network $G$.
\begin{lemma}\label{lmm-rerouteflow}
Given a confluent flow $h$ for routing all unit supplies in the $k$-layer network $H$ with node congestion $NC(h)$ and flow length $L(h)$, one can, in polynomial time, re-route $h$ to produce a confluent flow $f$ that routes a subset of supplies with value of at least a $\frac{1}{2k-1}$ fraction  of the total amount in $G$.
Furthermore, the node congestion and flow length of $f$ are bounded by $NC(h)$ and $k\cdot L(h)$, respectively.
\end{lemma}
\begin{proof}
We will use the following standard {\em conflict-free routing} technique to produce a confluent flow in $G$.

\begin{lemma}[\cite{shepherd2015polylogarithmic}]
Suppose $H$ is one $k$-layer node-capacitated network induced by $G$. Given a confluent flow $h$ in $H$ that satisfies all unit supplies $\{d_i\}_{i\in [\lambda]}$ and has node congestion $NC(h)$,  one can, in polynomial time, find a confluent flow $f$  in $G$ that routes a subset of supplies with value of at least a $\frac{1}{2k-1}$ fraction  of the total amount and has node congestion at most $NC(h)$.
\end{lemma}


Note that our multi-layer network differs from that in~\cite{shepherd2015polylogarithmic}. However,
the differences would not influence the re-routing of $h$. Indeed, when re-routing $h$, we can ignore the new introduced parameters in each arc, i.e., arc capacity and arc length. Meanwhile, the bi-directional vertical arcs (for undirected graphs) in each layer of $H$, as well as the fact that the appearance of vertical arcs depends on edge capacity rather than node capacity, would not bring in any difference. Because the re-routing of $h$ is based on the {\em support of $h$} rather than the constructed network $H$. Finally, we view those dummy sinks as the normal nodes in $H$.

By the above analysis, we can find a confluent flow $f$ in $G$ from the given confluent flow  $h$. Also, $f$ has node congestion at most $NC(h)$, and it can route a subset of supplies with value of at least a $\frac{1}{2k-1}$ fraction  of the total amount.


The only remaining thing is to bound the length of the resulting flow. Recall that the re-routing of $h$ is carried out on the support of $h$ (a directed rooted tree) plus all horizonal arcs. Without those horizonal arcs, the remaining support of $h$ are a collection of (directed) sub-trees in the $k$ corresponding layers. Also, because the flow length of $h$ is bounded by $L(h)$, any path in those sub-trees has its length bounded by $L(h)$ as well.
Note that the  conflict-free method finds the confluent flow $f$ by re-routing $h$, and the support of any path in $f$ in each layer is certain sub-path of $h$ in that  layer. Thus, any path in $f$ has length at most  $k\cdot L(h)$, since there are at most $k$ layers.
\end{proof}

On the other hand, note that Lemma~\ref{lmm-rerouteflow} can only bound the node congestion.
With the help of monotonic structure and dummy sinks, we bound the edge congestion as follow.
\begin{lemma}\label{lmm-edgecongestion}
The confluent flow $f$ in $G$ found  in Lemma~\ref{lmm-rerouteflow} has edge congestion at most $NC(h)$.
\end{lemma}
\begin{proof}
In Lemma~\ref{lmm-rerouteflow}, we bound the node congestion of the resulting confluent flow $f$. We would like to bound the edge congestion of $f$ by utilizing its node congestion.

Indeed, for a confluent flow, we have the following observations.
\begin{observation}\label{obs-1}
For a confluent flow in the uncapacitated network, the node congestion of any node $v$ is no less than the maximum edge congestion over all edges incident to $v$.
\end{observation}

In fact,  the confluent flow $f$ can induce a confluent flow $h'$ in $H$ by mapping a flow on one edge $e$ into its edge copy in $H$ with capacity that is larger than and nearest to the value $f(e)$. Note that the induced flow $h'$ does not use the identical vertical arc-copy in different layers. Clearly,
the node and edge congestion of $f$ are identical to the node and edge congestion of $h'$.
Thus, it suffices to consider the node and edge congestion of $h'$ in $H$.


Now let us focus on only a fixed layer of $H$, and consider the nodes and vertical arcs that the flow $h'$ passes through inside that layer. Due to the construction of the multi-layer graph, all vertical arcs and nodes in the same layer possess the same capacity, which we can view as uncapacitated network. Thus, by  Observation~\ref{obs-1}, we bound the edge congestion in each layer by the node congestion of $h'$. Regarding those horizontal arcs, since they connect different copies of nodes, it would not induce any edge in the original network $G$ and we don't need to consider its edge congestion.

Finally, we need to consider the congestion of the edge incident to the sink $t$ in $f$, say $(u,t)$. Suppose $(u,t)$ has capacity of $2^i$. Then, there are a copy $t_u^{j}$ of $t$ in the $j$-th layer, for each $j=i,i+1,...,k-1$, and a directed path $u^i,t_u^{i},t_u^{i+1},...,t_u^{k-1},t$, which is mapped to one arc $(u,t)$ in $f$. Note that those dummy sinks are not connected with any other node in $H$ except $u^i$ and $t$, and then there is no other in-flow to $t$ along this path.
Thus, this path contains flow of value the same as the arc $(u^i,t_u^{i})$. We  only need to consider the edge congestion of $(u^i,t_u^{i})$, which is the same as the edge congestion of $(u,t)$ in $f$. Clearly, the edge congestion of $(u^i,t_u^{i})$ is bounded by the node congestion of $h'$.

Since the node congestion of $h'$ remains the same as the node congestion of $h$ due to the re-routing scheme. Thus, the edge congestion of $f$ is also bounded by $NC(h)$.
\end{proof}

Putting all together, we can prove Theorem~\ref{thm-generallengthboundflow}.
\begin{proof}[Proof of Theorem~\ref{thm-generallengthboundflow}]
As described in the begin of Section~\ref{sec:StaticConFlow}, we pre-process the network and group together those  supplies of the same size, and then  consider routing each group separately via a confluent flow, outputting the best group as the solution. Note that this process will lose only an $O(\log \kappa)$ factor in the approximation guarantee. Also, note that the rounding up of capacities and supplies to the power of 2 would bring in a 2 factor to the edge congestion, because after rounding, certain supplies would be allowed to go through edges with smaller capacity.

Suppose there exists an  $L$-length-bounded confluent flow $f$ with edge congestion at most 1 for routing all supplies to the single sink in $G$.
Observe that $f$ induces an $L$-length-bounded confluent sub-flow $f'$ for routing all supplies in that group (uniform-supply case).
Then, we construct the multi-layer network $H$ for this supply group.
Also, notice that the such a confluent sub-flow  $f'$  induces a confluent $L$-length-bounded flow with edge congestion at most 1 in $H$.  Due to the construction of $H$, in each layer, edge capacity is the same as  node capacity, which can be viewed as uncapacitated network. Observe that a confluent flow in an uncapacitated network has node congestion the same as the edge congestion. Thus, the induced flow by $f'$ has node congestion at most 1.

Thus, there exists an $L$-length-bounded flow with node congestion at most 1 in $H$.
To route each group separately, we first
utilize  standard techniques\footnote{e.g., Theorems~2.5 and~2.8 in G.~Baier, {\em Flows with path restrictions}, PhD thesis, TU Berlin, 2003.}
to find out a feasible $O(1)L$-length-bounded (splittable) flow $\tilde{f}$ for routing all supplies to the sink.
Then,
applying Theorem~\ref{thm-klayerflow} immediately yields a confluent flow for routing all supplies in this group with node congestion at most $O(\log^8 n)$ and with length bounded by $O\left( L\log^2 n/\log\log n \right)$ whp. Finally, utilizing  Lemmas~\ref{lmm-rerouteflow} and~\ref{lmm-edgecongestion}, we can obtain a confluent flow for routing a subset of supplies, with value of at least a $\dfrac{1}{2k-1}=1/O(\log \kappa)$ faction of the total amount in this group. Furthermore, the found flow has edge congestion at most $O(\log^8 n)$ and length bounded by $O\left( L\log^3 n/\log\log n \right)$ whp.

Note that, together with the lose of approximation factor in grouping and rounding, this indeed gives us a confluent flow $\mathbbm{f}$ for routing a subset of supplies, with value of at least a $1/O(\log^2 \kappa)$ faction of the total amount in $G$, with the same edge congestion and flow length mentioned above. Thus, we complete the proof.
\end{proof}


\subsection{Polylogarithmic Bicriteria Approximation for  Confluent Dynamic Flows}
\label{apd:proof_main_thm}

This section proves the two main theorems for confluent dynamic flows.
First, we give the following lemmas for the transformation between static and dynamic flows.
\begin{lemma}\label{lmm-trans1}
A feasible unsplittable/confluent dynamic flow, which routes the supply $d_i$ from $s_i$ to $t$ ($\forall i\in[k]$) within time horizon $T$ , induces
a feasible unsplittable/confluent  $T$-length-bounded static flow, which routes the supply $d_i/T$ from $s_i$ to $t$ ($\forall i\in[k]$) in the same underlying network.
\end{lemma}
\begin{proof}
Suppose $f$ is the given dynamic flow for sending the supplies $\{d_i\}_{i\in [k]}$ located at $\{s_i\}_{i\in [k]}$. Let $f_e^i(t)$ be the value of the supply brought from $s_i$ by $f$ on edge $e$ at time $t$. Set $x_e^i:=\frac{1}{T}\sum_{t=1}^T f_e^i(t)$, and let $y_e=\sum_{i\in[k]}x_e^i$. Clearly, all $y_e$'s together induce a static flow, denoted as $f'$, in the same underlying network.

Since $f$ is feasible,
$$
\forall e\in E,\ \ \ \forall t\in\{0,1,...,T\},\ \ \ \sum_{i\in[k]}f_e^i(t)\leq c(e).
$$
This implies
$$\forall e\in E,\ \ \
y_e=\sum_{i\in[k]}\left( \frac{1}{T}\sum_{t=1}^T f_e^i(t) \right)=\frac{1}{T}\sum_{t=1}^T\left( \sum_{i\in[k]} f_e^i(t) \right)\leq c(e),
$$
which mean $f'$ is feasible.

Also, since we take the average of $f$ over time $T$ for each supply, $f'$ can send only the supply  $d_i/T$ from $s_i$ to the sink. Furthermore, when taking average of $f$ to produce $f'$, we still use the support of the dynamic flow $f$ to route $f'$. Hence, if $f$ is an  unsplittable/confluent dynamic flow, $f'$ must be a unsplittable/confluent static flow. Finally, since $f$ has time horizon $T$, any path inside the support of  $f$ must be $T$-length-bounded. Thus, we complete the proof.
\end{proof}

\begin{lemma}\label{lmm-trans2}
A feasible unsplittable/confluent $T$-length-bounded static flow, which routes the supply $d_i$ from $s_i$ to $t$ ($\forall i\in[k]$), induces a feasible unsplittable/confluent  dynamic flow, which routes the supply $z\cdot d_i$ from $s_i$ to $t$ ($\forall i\in[k]$) within time horizon $T+z$  in the same underlying network.
\end{lemma}
\begin{proof}
Suppose $f$ is  the given dynamic flow for sending the supplies $\{d_i\}_{i\in [k]}$ located at $\{s_i\}_{i\in [k]}$. Also, suppose $f$ can be  specified by a collection of source-sink paths $\mathcal{P}=(P_1,...,P_K)$ and corresponding flow values $f_1,...,f_K$.

Then, we let $f'$ be the dynamic flow that sends $f_j$ units of supplies along the path $P_j$ per unit time, for each $j\in[K]$. Clearly, $f'$ is feasible. Then, since each path $P_j$ is  $T$-length-bounded, the supply we send along it will reach the sink after $T$ units of time. Within time horizon $T+z$, $f'$ can totally send  at least $z f_j$ along $P_j$.  Recall that $f$ routes the supply $d_i$ from $s_i$ to $t$ ($\forall i\in[k]$). It means that $f'$ can route the supply $z\cdot d_i$ from $s_i$ to $t$ ($\forall i\in[k]$) within time horizon $T+z$.

On the other hand, note that we use the same support of $f$ to route $f'$. Thus, if  $f$ is an unsplittable/confluent static flow, $f'$ must be an unsplittable/confluent  dynamic flow. Thus, we complete the proof.
\end{proof}

Now, we start to prove Theorem~\ref{thm-PolylogConQuickFlow}.
\begin{proof}[Proof of Theorem~\ref{thm-PolylogConQuickFlow}]
Suppose $G$ is a directed/undirected network with a sink $t$ and a collection of supplies $\{d_i\}_{i=1}^\kappa$ located at sources $\{s_i\}_{i=1}^\kappa$.
Since $OPT$\footnote{This paper assumes that the optimal solution $OPT$ is bounded by $2^{O(n)}$.} is the optimal time, there exists a feasible confluent dynamic flow for routing all supplies $\{d_i\}_{i=1}^\kappa$ to the single sink with time horizon $OPT$. By Lemma~\ref{lmm-trans1}, we know it induces a feasible confluent $OPT$-length-bounded static flow for routing supplies $\{d_i/OPT\}_{i=1}^\kappa$.

Then, applying Theorem~\ref{thm-generallengthboundflow}, we can compute, in polynomial time, a confluent static flow for routing {\em a subset of supplies} with value
at least $\sum_{i\in[\kappa]}d_i/O(\log^{2} \kappa)$,
with edge congestion $O(\log^{8} n)$ and  flow length $O(\log^3 n/\log\log n)\cdot OPT$, whp. We scale the flow down by an $O(\log^{8} n)$ factor, which ensures the resulting flow is feasible
and allows it to send to the sink $t$ a subset of supplies among $\{d_i/(O(\log^{8} n)\cdot OPT)\}_{i=1}^\kappa$, with value at least a $1/O(\log^{2} \kappa)$ fraction of the total amount.
%
%
%
Now, by Lemma~\ref{lmm-trans2}, we can change the found flow into a feasible confluent dynamic flow that routes {\em a subset of supplies} with value
at least $\sum_{i\in[\kappa]}d_i/O(\log^{2} \kappa)$ to $t$,
$$
O(\log^3 n/\log\log n)\cdot OPT+O(\log^{8} n)\cdot OPT=O(\log^{8} n)\cdot OPT.
$$

However, $OPT$ is not explicitly known. To deal with this issue, we use a binary search to determine $OPT$. Initially, we set $T_1=0$ and $T_2\geq OPT$. Start from  $T=(T_1+T_2)/2$ and find a confluent dynamic flow by the above process. Clearly, if $T\geq OPT$, there is a feasible confluent dynamic flow to route all supplies with time $T$, and we can find out a bicriteria approximation of such a flow; if $T< OPT$, we fail and return no flow. When we succeed in constructing an approximation, we set $T_2=T$ and continue the binary search; otherwise, we set $T_1=T$ and repeat. Thus, at most $O(\log (OPT))$ times of computation suffice to determine $OPT$. This gives the polynomial-time algorithm for finding the desired bicriteria approximation of the single-sink Confluent Quickest Flow problem whp.
\end{proof}

Similarly, we can prove Theorem~\ref{thm-PolylogConFlowOverTime}.
\begin{proof}[Proof of Theorem~\ref{thm-PolylogConFlowOverTime}]
Suppose $G$ is a directed/undirected network with a sink $t$ and a collection of sources $\{s_i\}_{i=1}^\kappa$.
Since $OPT$ is the optimal amount of supplies that can be confluently sent to $t$ within the time horizon $T$, there exists a feasible confluent dynamic flow for routing $OPT$ units of supplies to the sink within the time horizon $T$.

By Lemma~\ref{lmm-trans1}, we know it induces a feasible confluent  $T$-length-bounded static flow for routing totally $OPT/T$ units of supplies to $t$. By  standard techniques
\footnote{e.g., Theorem~2.8 in G.~Baier, {\em Flows with path restrictions}, PhD thesis, TU Berlin, 2003.}
one can find out a splittable $O(1)T$-length-bounded flow $f$ with the total flow amount at least as large as the maximum $T$-length-bounded  flow. Note that the total value of $f$ is at least $OPT/T$. Suppose $f$ sends $d_i$ units of supplies from $s_i$ to $t$, for each $i\in[\kappa]$, with $\sum_{i}d_i\geq OPT/T$.
Hence, there exists a feasible splittable  $O(1)T$-length-bounded static flow for routing all supplies $\{d_i\}_{i=1}^\kappa$. Applying Theorem~\ref{thm-generallengthboundflow}, we can compute, in polynomial time, a confluent static flow for routing at least a $1/O(\log^{2} \kappa)$ fraction of supplies $\{d_i\}_{i=1}^\kappa$, with edge congestion bounded by $O(\log^{8} n)$ and flow length  bounded by $O(\log^3 n/\log\log n)\cdot T$, whp. Scaling the flow down by a $O(\log^{8} n)$ factor yields a feasible confluent flow. Now, we have a feasible $O(\log^3 n/\log\log n)\cdot T$-length-bounded confluent static flow  that routes at least a $1/O(\log^{2} \kappa)$ fraction of supplies $\{d_i/O(\log^{8} n)\}_{i=1}^\kappa$ whp.

Thus, by Lemma~\ref{lmm-trans2}, we can change the found flow into a feasible confluent dynamic flow $f'$ that routes at least a $1/O(\log^{2} \kappa)$ fraction of supplies $\{d_i T\}_{i=1}^\kappa$ to $t$, and whp  the time horizon is bounded by
$$
O(\log^3 n/\log\log n)\cdot T+O(\log^{8} n)\cdot T=O(\log^{8} n)\cdot T.
$$
Since $\sum_{i}d_i T\geq OPT$, $f'$ is a $(O(\log^2 \kappa),O(\log^8 n))$-approximation for  the single-sink Confluent Maximum Flow Over Time problem.
Since the value of $OPT$ is unknown, we again utilize the binary search to determine it. Thus, there exists a polynomial-time algorithm that, whp, computes the desired bicriteria approximation of the single-sink Confluent Maximum Flow Over Time problem.
\end{proof}

We now restrict our technique to static networks, and prove Theorem~\ref{thm-StaticConfluent}.
\begin{proof}[Proof of Theorem~\ref{thm-StaticConfluent}]
We start from the directed, edge-capacitated static network $G(V,A)$. Suppose $G(V,A)$ has $\kappa$ supplies $\{d_i\}_{i=1}^\kappa$. We pre-process the network into a multi-layer network, the same as Section~\ref{sec:StaticConFlow} except that we ignore arc lengths. Also,
we group together those  supplies of the same size, and then  consider routing each group separately via a confluent flow, outputting the best group as the solution. Note that this process will lose only a $O(\log \kappa)$ factor in the approximation guarantee.

Since each group has uniform supplies, we can view all supplies as unit supplies and  all capacities are integral. Hence we can use the standard technique to find an integral maximum flow, which upper bounds the maximum flow for routing supplies in a group. Then, because the constructed multi-layer network is a (node-capacitated) monotonic network with a single sink, satisfying the no-bottleneck assumption, we utilize Theorem~\ref{thm-monotonicflow-relaxed} to find out a confluent flow in the multi-layer network. Recall that Lemma~\ref{lmm-rerouteflow} works for the routing of static flows without the length-bounded constraint.
Then, by the analysis similar to Theorem~\ref{thm-generallengthboundflow}, we show one can find a confluent flow $f'$ for routing a subset of supplies to the sink, with edge congestion at most $O(\log^8 \kappa)$ whp. The total value of those supplies is at least $1/O(\log^2 \kappa)$ of $OPT$.

Scaling $f'$ down by a $O(\log^8 \kappa)$ factor  makes it a feasible flow that sends at least a $1/O(\log^8 \kappa)$ fraction of those supplies.
Then, we apply the following theorem:
\begin{theorem}[\cite{Chekuri2007}]
Let $T$ be an edge-capacitated tree instance with the NBA. Suppose there is a
fractional flow that routes a fraction $\gamma_i \in [0, 1]$ of each commodity $i$. Then a subset of the items can be found, in polynomial time, that feasibly routes a total demand of $\Omega(\sum_i \gamma_i d_i)$ on $T$.
\end{theorem}
This give us a confluent flow that routes a subset of supplies whose total value is at least $1/O(\log^{10} \kappa)$ of $OPT$.

For directed, node-capacitated networks, we only need to induce edge capacities: For each edge $e=(u,v)$, we let edge capacity $c(e)=\min(c(u),c(v))$. Then, we utilize the method above for edge-capacitated network to compute a confluent flow with edge congestion bounded. Note that, in such a confluent flow, the node congestion of a node $u$ is bounded by the congestion of edge that carried all flows out of $u$. This gives the desired approximation.

Undirected networks can be dealt with by replacing each edge by two opposite directional edges. Thus we complete the proof.
\end{proof}





\end{document}